\documentclass[usletter]{article}
\usepackage[totalwidth=450pt,totalheight=640pt]{geometry}
\usepackage{scrextend}
\newcommand{\longversion}[1]{#1}
\newcommand{\shortversion}[1]{}
\newcommand{\longshort}[2]{\longversion{#1}\shortversion{#2}}

\usepackage[cmex10]{amsmath}
\usepackage{amssymb}
\usepackage{amsthm,todonotes}
\usepackage{graphicx}
\usepackage{microtype}
\usepackage{color}
\usepackage{float}
\usepackage{enumitem}

\newcommand{\fo}{\mathcal{FO}}
\newcommand{\rela}{\mathbf{A}}

\newcommand{\cc}{\mathbf{C}}
\newcommand{\dd}{\mathbf{D}}
\newcommand{\pp}{\mathbf{P}}
\newcommand{\qq}{\mathbf{Q}}
\newcommand{\suchthat}{\;\ifnum\currentgrouptype=16 \middle\fi|\;}

\newtheorem{proposition}{Proposition}
\newtheorem{theorem}{Theorem}
\newtheorem{lemma}{Lemma}
\newtheorem{claim}{Claim}
\newtheorem{example}{Example}
\newtheorem{observation}{Observation}

\begin{document}
\shortversion{

\setlength{\pdfpageheight}{\paperheight}
\setlength{\pdfpagewidth}{\paperwidth}

\conferenceinfo{CSL-LICS 2014}{July 14--18, 2014, Vienna, Austria}
\copyrightyear{2014}
\copyrightdata{978-1-4503-2886-9}
\doi{nnnnnnn.nnnnnnn}





\titlebanner{banner above paper title}        
\preprintfooter{short description of paper}   
}

\title{Model Checking Existential Logic on Partially Ordered Sets\footnote{This research was supported by ERC Starting Grant (Complex Reason, 239962) 
and FWF Austrian Science Fund (Parameterized Compilation, P26200).}}

\longshort{
\newcommand*\samethanks[1][\value{footnote}]{\footnotemark[#1]}
\author{%
Simone Bova, Robert Ganian, and Stefan Szeider\\
\small Vienna University of Technology\\[-3pt]
\small  Vienna, Austria}

\date{}

}
{
\authorinfo{Simone Bova, Robert Ganian, and Stefan Szeider}
           {Vienna University of Technology (Vienna, Austria)}}

\maketitle

\begin{abstract}
We study the problem of checking whether an existential sentence 
(that is, a first-order sentence in prefix form built using existential quantifiers and all Boolean connectives) 
is true in a finite partially ordered set (in short, a poset). A poset is a reflexive, antisymmetric, and transitive digraph.  
The problem encompasses the fundamental embedding problem of finding an isomorphic 
copy of a poset as an induced substructure of another poset.  

Model checking existential logic is already $\textup{NP}$-hard on a fixed poset; 
thus we investigate structural properties of posets yielding 
conditions for fixed-parameter tractability when the problem is parameterized by the sentence. 
We identify width as a central structural property 
(the width of a poset is the maximum size of a subset of pairwise incomparable elements); 
our main algorithmic result is that model checking existential logic on 
classes of finite posets of bounded width is fixed-parameter tractable.  
We observe a similar phenomenon in classical complexity, 
where we prove that the isomorphism problem is polynomial-time tractable 
on classes of posets of bounded width; this settles an open problem in order theory.

We surround our main algorithmic result with complexity results 
on less restricted, natural neighboring classes of finite posets, 
establishing its tightness in this sense.  
We also relate our work with (and demonstrate its independence of) 
fundamental fixed-parameter tractability results for model checking on digraphs of bounded degree 
and bounded clique-width.  
\end{abstract}

\shortversion{


}

\section{Introduction}

\noindent \textit{Motivation.}  The \emph{model checking problem}, 
to decide whether a given logical sentence is true in a given structure, 
is a fundamental computational problem which appears in a variety of areas in
computer science, including database theory, artificial intelligence, constraint satisfaction, 
and computational complexity.  The problem is computationally intractable in its general version, 
and hence it is natural to seek restrictions of the class of structures or the class of sentences 
yielding sufficient or necessary conditions for computational tractability.  

Here, as usual in the complexity investigation of the model checking problem, 
computational tractability refers to \emph{polynomial-time tractability} or, 
in cases where polynomial-time tractability is unlikely, 
a relaxation known as \emph{fixed-parameter tractability with the sentence as a parameter}. 
The latter guarantees a decision algorithm running 
in $f(k) \cdot n^c$ time on inputs of size $n$ 
and sentences of size $k$, 
where $f$ is a computable function 
and $c$ is a constant. 
For further discussion of the complexity setup adopted here, 
including its algorithmic motivations, 
we refer the reader to \cite{Grohe07a, FlumGrohe06}.

The study of model checking first-order logic on restricted 
classes of finite \emph{combinatorial structures} is an established line of research originating from the seminal work of Seese \cite{Seese96}. Results in this area have provided very general conditions for computational tractability, and 
even exact characterizations in many relevant cases \cite{GroheKreutzerSiebertz14}.  
As Grohe observes \cite{Grohe07a}, though, 
it would be also interesting to investigate structural properties 
facilitating the model checking problem in the realm of finite \emph{algebraic structures}, 
for instance groups or lattices.  

In this paper, we investigate the class of finite \emph{partially ordered sets}.  
A partially ordered set (in short, a \emph{poset}) is the structure obtained by equipping a nonempty 
set with a reflexive, antisymmetric, and transitive binary relation.  In other words, 
the class of posets coincides with the class of directed graphs satisfying 
a certain universal first-order sentence (axiom); namely, the sentence that enforces 
reflexivity, antisymmetry, and transitivity of the edge relation.  
In this sense, from a logical perspective, posets form 
an intermediate case between combinatorial and algebraic structures; 
they can be viewed as being stronger than purely combinatorial structures, 
as the nonlogical vocabulary is presented by a first-order axiomatization; 
but weaker than genuinely algebraic structures, as the axiomatization 
is expressible in universal first-order logic (too weak of a fragment to define algebraic operations).  

Posets are fundamental combinatorial objects \cite[Chapter~8]{GrahamGrotschelLovasz95}, 
with applications in many fields of computer science, ranging from software verification \cite{NielsonNielsonHankin05} 
to computational biology \cite{RauschReinert10}.  However, 
very little is known about the complexity of the model checking problem on classes of finite posets; 
to the best of our knowledge, even the complexity of natural syntactic fragments of first-order logic 
on basic classes of finite posets is open.  

A prominent logic in first-order model-checking is \emph{primitive positive} logic, 
that is, first-order sentences built using existential quantification ($\exists$) 
and conjunction ($\wedge$); 
the problem of model checking primitive positive logic is 
equivalent to the \emph{constraint satisfaction problem} and the \emph{homomorphism problem} \cite{FederVardi98}.  However, 
restricted to posets, the problem of model checking primitive positive logic 
and even \emph{existential positive} logic, obtained from primitive positive logic by including disjunction ($\vee$) in the logical vocabulary, is trivial; 
because of reflexivity, every existential positive sentence is true on every poset!

As we observe (Proposition~\ref{pr:exprcomplex}), the complexity scenario changes abruptly 
in \emph{existential conjunctive} logic, that is, 
first-order sentences in prefix negation normal form built using 
$\exists$, $\wedge$, and negation ($\neg$).  Here, 
the model checking problem is $\textup{NP}$-hard 
even on a certain fixed finite poset; in the complexity jargon, 
the \emph{expression} complexity of 
existential conjunctive logic 
is $\textup{NP}$-hard on finite posets.   
In other words, as long as computational tractability is identified with 
polynomial-time tractability, any structural property of posets is 
algorithmically immaterial (in a sense that can be made precise).  There is then a natural quest 
for relaxations of polynomial-time tractability yielding 
\textit{(i)} a nontrivial complexity analysis of the problem, 
and \textit{(ii)} a refined perspective on the structural properties of posets 
underlying tamer algorithmic behaviors; in this paper 
we achieve \textit{(i)} and \textit{(ii)} through the glasses of fixed-parameter tractability.  

More precisely, as we discuss below, our contribution is a complete description of 
the parameterized complexity of model checking (all syntactic fragments of) existential first-order logic 
(first-order sentences in prefix normal form built using $\exists$, $\wedge$, $\vee$, and $\neg$), 
with respect to classes of finite posets in a hierarchy 
generated by fundamental poset invariants.\footnote{Note that existential \emph{disjunctive} logic (first-order sentences in prefix negation normal form 
built using $\exists$, $\vee$, and $\neg$) is trivial on posets.  
In fact, every sentence in the fragment is either true on every poset, or false on every poset, 
and it is easy to check which of the two cases holds for any given sentence.}  

Model checking existential logic encompasses as a special case the fundamental 
\emph{embedding problem}, to decide whether a given structure contains an isomorphic 
copy of another given structure as an \emph{induced} substructure; 
in fact, the embedding problem reduces in polynomial-time to the problem of 
model checking certain existential (even conjunctive) sentences.  The aforementioned fact 
that existential conjunctive logic is already $\textup{NP}$-hard on a fixed finite poset 
leaves open the existence of a nontrivial classical complexity classification of the embedding problem. 
We provide such a classification by giving a complete description of 
the classical complexity of the embedding problem in the introduced hierarchy of poset invariants.  

%

We hope that the investigation of the existential fragment prepares the ground 
(and possibly provides basic tools) for understanding the model checking problem 
for more expressive logics on posets.

\medskip

\noindent \textit{Contribution.}  We now give an account of our contribution.  
We refer the reader to Figure~\ref{fig:overwparcompl} for an overview; 
the poset invariants and their relations are introduced in Section~\ref{sect:setup}.  

\begin{figure}[ht]
\centering
\begin{picture}(0,0)%
\includegraphics{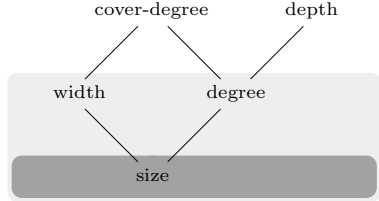}%
\end{picture}%
\setlength{\unitlength}{2279sp}%
\begingroup\makeatletter\ifx\SetFigFont\undefined%
\gdef\SetFigFont#1#2#3#4#5{%
  \reset@font\fontsize{#1}{#2pt}%
  \fontfamily{#3}\fontseries{#4}\fontshape{#5}%
  \selectfont}%
\fi\endgroup%
\begin{picture}(4074,2315)(2014,-10243)
\put(2521,-9106){\makebox(0,0)[lb]{\smash{{\SetFigFont{7}{8.4}{\rmdefault}{\mddefault}{\updefault}{\color[rgb]{0,0,0}$\textup{width}$}%
}}}}
\put(3421,-10006){\makebox(0,0)[lb]{\smash{{\SetFigFont{7}{8.4}{\rmdefault}{\mddefault}{\updefault}{\color[rgb]{0,0,0}$\textup{size}$}%
}}}}
\put(4186,-9106){\makebox(0,0)[lb]{\smash{{\SetFigFont{7}{8.4}{\rmdefault}{\mddefault}{\updefault}{\color[rgb]{0,0,0}$\textup{degree}$}%
}}}}
\put(2971,-8206){\makebox(0,0)[lb]{\smash{{\SetFigFont{7}{8.4}{\rmdefault}{\mddefault}{\updefault}{\color[rgb]{0,0,0}$\textup{cover\textup{-}degree}$}%
}}}}
\put(5041,-8206){\makebox(0,0)[lb]{\smash{{\SetFigFont{7}{8.4}{\rmdefault}{\mddefault}{\updefault}{\color[rgb]{0,0,0}$\textup{depth}$}%
}}}}
\end{picture}%
\caption{
The (light or dark) gray region covers invariants such that, 
if a class of finite posets is bounded under the invariant, 
then 
model checking existential logic (or equivalently, by Proposition~\ref{proposition:metodological}, 
model checking existential conjunctive logic, or deciding embedding)
over the class is fixed-parameter tractable; 
the white region covers invariants such that 
there exists a class of finite posets bounded under the invariant 
where the problem is $\textup{W}[1]$-hard.  Similarly, the dark gray region covers invariants where 
the embedding problem is polynomial-time tractable, 
and the complement of the dark gray region (light gray or white) covers invariants where the problem is $\textup{NP}$-hard.  
In classical complexity, as opposed to parameterized complexity, 
the tractability frontier of existential (conjunctive) logic and embedding 
are different (the former, since existential logic is already $\textup{NP}$-hard on a fixed finite poset, 
is $\textup{NP}$-hard everywhere).}
\label{fig:overwparcompl}
\end{figure}

In contrast to the classical case, model checking existential logic on fixed structures 
is trivially fixed-parameter tractable; in fact, even the full first-order logic is trivially fixed-parameter tractable on 
any class of finite structures of bounded size. On the other hand, there exist 
classes of finite posets where existential logic 
is unlikely to be fixed-parameter tractable (in fact, 
there exist classes where even the embedding problem is $\textup{W}[1]$-hard); 
but the reduction class given by the natural hardness proof is rather wild, 
in particular it has bounded depth but \emph{unbounded width} (Proposition~\ref{pr:parhardallposets}).

The \emph{width} of a poset is the maximum size of a subset of pairwise incomparable elements (antichain); 
along with its \emph{depth}, the maximum size of a subset of pairwise comparable elements (chain), 
these two invariants form the basic and fundamental structural properties of a poset, 
arguably its most prominent and natural features.  Our main result establishes that 
width helps algorithmically (in contrast to depth); specifically, 
we prove that \emph{model checking existential logic on classes of finite posets 
of bounded width is fixed-parameter tractable} (Theorem~\ref{thm:EFOFPT}).  
This, together with Seese's  algorithm (plus a routine reduction described in Proposition~\ref{pr:harddegree}), allows us 
to complete the parameterized complexity classification of the investigated poset invariants, 
as depicted in Figure~\ref{fig:overwparcompl}.

We believe that our tractability result essentially enlightens the fundamental feature of posets of bounded width 
that can be exploited algorithmically; namely, \emph{bounded width posets admit a polynomial-time compilation to certain semilattice structures}, which 
are algorithmically tamer than the original posets, but equally expressive with respect to the problem at hand.  
The proof proceeds in two stages.  We first prove that, on any class of finite relational structures, 
model checking existential logic is fixed-parameter tractable if and only 
if the embedding problem is fixed-parameter tractable (Proposition~\ref{proposition:metodological}). Next, 
using 
the color coding technique of Alon, Yuster, and Zwick \cite{AlonYusterZwick95}, 
we reduce an instance of the embedding problem on posets of bounded width to 
a suitable family of instances of the homomorphism problem of certain semilattice structures, 
which is polynomial-time tractable by classical results of Jeavons, Cohen, and Gyssens \cite{JeavonsCohenGyssens97}.   

Our approach is reminiscent of the well established fact in order theory 
that finite posets correspond exactly (in a sense that can be made precise in category-theoretic terms) 
to finite distributive lattices.  However, 
the algorithmic implications of this correspondence 
have been possibly overlooked.  Indeed, using the correspondence and the known fact that 
the isomorphism problem is polynomial-time tractable on finite distributive lattices, 
we prove that \emph{the isomorphism problem for posets of bounded width is polynomial-time tractable} (Theorem~\ref{th:isoptime}), 
which 
settles an 
open question in order theory \cite[p.~284]{CaspardLeclercMonjardet12}.

Motivated by the equivalence (in parameterized complexity) between embedding and model checking existential conjunctive logic 
(Proposition~\ref{proposition:metodological}) on one hand, 
and the fact that existential conjunctive logic is already $\textup{NP}$-hard on a fixed finite poset (Proposition~\ref{pr:exprcomplex}) on the other hand, 
we also revisit the classical complexity of the embedding problem for finite posets and classify 
it with respect to the poset invariants studied in the parameterized complexity setting.  
The outcome is pictured in Figure~\ref{fig:overwparcompl}; here, polynomial-time tractability 
of the embedding problem on posets of bounded size is optimal with respect to the studied poset 
invariants.  We remark that the hardness results are technically involved (Theorem~\ref{th:widthnphard} and Theorem~\ref{th:degreenphard}); 
in particular, bounded width is a known obstruction for hardness proofs 
(for instance, the complexity of the dimension problem is unknown on bounded width posets).

We conclude mentioning that our work on posets relates with, but is independent of, 
general results by Seese \cite{Seese96} and Courcelle, Makowsky, and Rotics \cite{CourcelleMakowskyRotics00}, respectively, on model checking first-order logic 
on classes of finite graphs of bounded degree and bounded clique-width.  Namely, the order relation of a poset has bounded degree 
if and only if the poset has bounded depth and bounded \emph{cover-degree} (that is, its cover relation has bounded degree); 
moreover, if a poset has bounded width, then it has bounded cover-degree (Proposition~\ref{pr:diagram}).  
However, \emph{there exist classes of bounded width posets with unbounded degree} (for instance, chains), 
and \emph{there exist classes of bounded width posets with unbounded clique-width} (Proposition~\ref{pr:nocw}), 
which excludes the direct application of the aforementioned results.

\medskip

\noindent \shortversion{\emph{Throughout the paper, we mark with $\star$ all statements whose proofs 
are omitted; we refer to the \href{http://arxiv.org/abs/???}{arXiv} for a full version.}}

\section{Preliminaries}\label{sect:prelim}

For all integers $k \geq 1$, we let $[k]$ denote the set $\{ 1, \ldots, k \}$.

\medskip

\noindent \textit{Logic.} In this paper, we focus on relational first-order logic.  
A \emph{vocabulary} $\sigma$ is a \emph{finite} set of \emph{relation symbols}, 
each of which is associated to a natural number called its
\emph{arity}; we let $\textup{ar}(R)$ denote the arity of $R \in \sigma$.  
An \emph{atom} $\alpha$ (over vocabulary $\sigma$) is an equality of variables ($x=y$) 
or is a predicate application $R x_1 \dots x_{\textup{ar}(R)}$, 
where $R \in \sigma$ and $x_1,\dots,x_{\textup{ar}(R)}$ are variables.  
A \emph{formula} (over vocabulary $\sigma$) is built from atoms (over $\sigma$), 
conjunction ($\wedge$), disjunction ($\vee$), negation ($\neg$), 
universal quantification ($\forall$), and existential quantification ($\exists$).  
A \emph{sentence} is a formula having no free variables.  We let $\fo$ 
denote the class of first-order sentences in \emph{prefix negation normal form}, 
that is, for each $\phi \in \fo$, 
the quantifiers occur in front of the sentence and the negations occur in front of the atoms.  

Let $\rho$ be a subset of $\{\forall,\exists,\wedge,\vee,\neg\}$ 
containing at least one quantifier and at least one binary connective.  
We let $\fo(\rho) \subseteq \fo$ denote the \emph{syntactic fragment} of $\fo$-sentences 
built using only logical symbols in $\rho$.  We call 
$\fo(\exists,\wedge,\vee,\neg)$ the \emph{existential} fragment, 
$\fo(\exists,\wedge,\neg)$ the \emph{existential conjunctive} fragment, 
and $\fo(\exists,\wedge)$, the \emph{existential conjunctive positive} (or \emph{primitive positive}) fragment.  
 
\medskip

\noindent \textit{Structures.}  
Let $\sigma$ be a relational vocabulary.  A \emph{structure} $\rela$ (over $\sigma$) is specified by 
a nonempty set $A$, called the \emph{universe} of the structure, 
and a relation $R^{\rela} \subseteq A^{\textup{ar}(R)}$ for each relation symbol $R \in \sigma$.  
A structure is \emph{finite} if its universe is finite.  

\emph{All structures considered in this paper are finite.} 

Given a structure $\mathbf{A}$ 
and $B \subseteq A$, we denote by $\mathbf{A}|_B$ the substructure of $\mathbf{A}$ induced by $B$, 
namely the universe of $\mathbf{A}|_B$ is $B$ and $R^{\mathbf{A}|_B}=R^{\mathbf{A}} \cap B^{\textup{ar}(R)}$ 
for all $R \in \sigma$.

Let $\mathbf{A}$ and $\mathbf{B}$ be $\sigma$-structures.  
A \emph{homomorphism} from $\mathbf{A}$ to $\mathbf{B}$ 
is a function $h \colon A \to B$ such that $(a_1,\ldots,a_{\textup{ar}(R)}) \in R^{\mathbf{A}}$ 
implies $(h(a_1),\ldots,h(a_{\textup{ar}(R)})) \in R^{\mathbf{B}}$, 
for all $R \in \sigma$ and all $(a_1,\ldots,a_{\textup{ar}(R)}) \in A^{\textup{ar}(R)}$; 
a homomorphism from $\mathbf{A}$ to $\mathbf{B}$ is \emph{strong} 
if $(a_1,\ldots,a_{\textup{ar}(R)}) \not\in R^{\mathbf{A}}$ 
implies $(h(a_1),\ldots,h(a_{\textup{ar}(R)})) \not\in R^{\mathbf{B}}$.  
An \emph{embedding} from $\mathbf{A}$ to $\mathbf{B}$ 
is an injective strong homomorphism from $\mathbf{A}$ to $\mathbf{B}$.  
An \emph{isomorphism} from $\mathbf{A}$ to $\mathbf{B}$ 
is a bijective embedding from $\mathbf{A}$ to $\mathbf{B}$.  

\emph{In graph theory, an injective strong homomorphism is also called a \lq\lq strong embedding\rq\rq,  
and the term \lq\lq embedding\rq\rq\ is used in the weaker sense of injective homomorphism; here, 
we adopt the order-theoretic (and model-theoretic) terminology.} 

For a structure $\rela$ and a sentence $\phi$ over the same vocabulary, 
we write $\rela \models \phi$ if the sentence $\phi$ is \emph{true} in the structure $\rela$.  
When $\rela$ is a structure, $f$ is a mapping from variables to 
the universe of $\rela$, and $\psi(x_1,\ldots,x_n)$ is a formula over the vocabulary of $\rela$,
we liberally write $\rela \models \psi(f(x_1),\ldots,f(x_n))$ to indicate that $\psi$ is satisfied 
by $\rela$ and $f$.  

%

A structure $\mathbf{G}=(V,E^\mathbf{G})$ 
with $\textup{ar}(E)=2$ is called a \emph{digraph}, 
and a \emph{graph} if $E^\mathbf{G}$ is irreflexive and symmetric.  We let $\mathcal{G}$ 
denote the class of all graphs. Let $\mathbf{G}$ be a digraph.  The \emph{degree} of $g \in G$, in symbols $\textup{degree}(g)$,  
is equal to $|\{ (g',g) \in E^\mathbf{G} \mid g' \in G \} 
\cup \{ (g,g') \in E^\mathbf{G} \mid g' \in G \}|$, 
and the \emph{degree} of $\mathbf{G}$, in symbols $\textup{degree}(\mathbf{G})$, 
is the maximum degree attained by the elements of $\mathbf{G}$.  

A digraph $\pp=(P,\leq^\pp)$ 
is a \emph{poset} if $\leq^\pp$ is a \emph{reflexive}, 
\emph{antisymmetric}, and \emph{transitive} relation over $P$, 
that is, respectively, $\pp \models \forall x(x \leq x)$, 
$\pp \models \forall x \forall y((x \leq y \wedge y \leq x) \to x=y)$, 
and $\pp \models \forall x \forall y \forall z((x \leq y \wedge y \leq z) \to x \leq z)$.  

A \emph{chain} in $\pp$ is a subset $C \subseteq P$ 
such that $p \leq^{\pp} q$ or $q \leq^{\pp} p$ for all $p,q \in C$ 
(in particular, if $P$ is a chain in $\pp$, 
we call $\pp$ itself a chain).  
We say that $p$ and $q$ are \emph{incomparable} 
in $\pp$ (denoted $p \parallel^{\pp} q$) if 
$\pp \not\models p \leq q \vee q\leq p$.  
An \emph{antichain} in $\pp$ is a subset $A \subseteq P$ 
such that $p \parallel^{\pp} q$ for all $p,q \in A$ 
(in particular, if $P$ is an antichain in $\pp$, 
we call $\pp$ itself an antichain).  

Let $\pp$ be a poset and let $p,q \in P$.  We say that $q$ \emph{covers} $p$ in $\pp$ (denoted $p \prec^{\pp} q$) 
if $p<^{\pp}q$ and, for all $r \in P$, $p \leq^{\pp} r <^{\pp}q$ implies $p=r$.  
The \emph{cover graph} of $\pp$ is the digraph $\textup{cover}(\pp)$ with vertex set $P$ 
and edge set $\{ (p,q) \mid p \prec^{\pp} q \}$.  
If $\mathcal{P}$ 
is a class of posets, we let $\textup{cover}(\mathcal{P})=\{ \textup{cover}(\pp) \mid \pp \in \mathcal{P} \}$. 
\longshort{It is well known that 
computing the cover relation corresponding to a given order relation, 
and vice versa the order relation corresponding to a given cover relation, 
is feasible in polynomial time \cite{Schroder03}.}{It is well known that 
computing the cover relation corresponding to a given order relation, 
and vice versa the order relation corresponding to a given cover relation, 
is feasible in polynomial time \cite{Schroder03}.}

In the figures, 
posets are represented by their \emph{Hasse diagrams}, 
that is a diagram of their cover relation 
where all edges are intended oriented upwards.  


Let $\mathcal{P}$ be the class of all posets.  A \emph{poset invariant} 
is a mapping $\textup{inv} \colon \mathcal{P} \to \mathbb{N}$ such that $\textup{inv}(\mathbf{P})=\textup{inv}(\mathbf{Q})$ 
for all $\mathbf{P},\mathbf{Q} \in \mathcal{P}$ such that $\mathbf{P}$ and $\mathbf{Q}$ are isomorphic.  
Let $\textup{inv}$ be any invariant over $\mathcal{P}$.  Let $\mathcal{P}$ be any class of posets.  
We say that $\mathcal{P}$ is \emph{bounded} with respect to $\textup{inv}$ if there exists $b\in \mathbb{N}$ such that 
$\textup{inv}(\mathcal{P})\leq b$ for all $\mathbf{P} \in \mathcal{P}$.  
Two poset invariants are \emph{incomparable} if there exists a class of posets 
bounded under the first but unbounded under the second, and there exists a class of posets  bounded under the second but unbounded under the first.


\medskip

\noindent \textit{Problems.}  We refer the reader to \cite{FlumGrohe06} for 
the standard algorithmic setup of the model checking problem, 
including the underlying computational model, 
encoding conventions for input structures and sentences, 
and the notion of \emph{size} of the (encoding of an) input structure or sentence.  
We also refer the reader to \cite{FlumGrohe06} for further background in parameterized complexity theory 
(including the notion of \emph{fpt many-one reduction} and \emph{fpt Turing reduction}).

Here, we mention that a \emph{parameterized problem} $(Q,\kappa)$ is a 
\emph{problem} $Q \subseteq \Sigma^*$ together 
with a \emph{parameterization} $\kappa \colon \Sigma^* \to \mathbb{N}$, where $\Sigma$ is a finite alphabet.  
A parameterized problem $(Q,\kappa)$ is \emph{fixed-parameter tractable (with respect to $\kappa$)}, 
in short \emph{fpt}, if there exists a decision algorithm for $Q$, 
a computable function $f \colon \mathbb{N} \to \mathbb{N}$, 
and a polynomial function $p \colon \mathbb{N} \to \mathbb{N}$, 
such that for all $x \in \Sigma^*$, the running time of the algorithm on $x$ 
is at most $f(\kappa(x)) \cdot p(|x|)$.  
%
%
We provide evidence that 
a parameterized problem is not fixed-parameter tractable 
by proving that the problem is $\textup{W}[1]$-hard under fpt many-one reductions; 
this holds unless the exponential time hypothesis fails \cite{FlumGrohe06}.

The (parameterized) computational problems under consideration are the following.  
Let $\sigma$ be a relational vocabulary, 
$\mathcal{C}$ be a class of 
$\sigma$-structures, 
and $\mathcal{L} \subseteq \fo$ be a class of $\sigma$-sentences.  
The \emph{model checking problem} for $\mathcal{C}$ and $\mathcal{L}$, 
in symbols $\textsc{MC}(\mathcal{C},\mathcal{L})$, is the problem of deciding, 
given $(\mathbf{A},\phi) \in \mathcal{C} \times \mathcal{L}$, 
whether $\mathbf{A} \models \phi$.  The parameterization, given an instance $(\mathbf{A},\phi)$, 
returns the size of the encoding of $\phi$.  The \emph{embedding problem} 
for $\mathcal{C}$, 
in symbols $\textsc{Emb}(\mathcal{C})$, 
is the problem of deciding, given a pair $(\mathbf{A},\mathbf{B})$, 
where $\mathbf{A}$ is a $\sigma$-structure and 
$\mathbf{B}$ is a $\sigma$-structure in $\mathcal{C}$, 
whether $\mathbf{A}$ embeds into $\mathbf{B}$. 
The parameterization, given an instance $(\mathbf{A},\mathbf{B})$, 
returns the size of the encoding of $\mathbf{A}$.  The problems $\textsc{Hom}(\mathcal{C})$ 
and $\textsc{Iso}(\mathcal{C})$ are defined similarly 
in terms of homomorphisms and isomorphisms respectively.

\section{Basic Results}\label{sect:setup}

In this section, we set the stage for our parameterized and classical 
complexity results in Section~\ref{sect:mainresults} and Section~\ref{sect:classical} respectively.  
We start observing some basic reducibilities between the problems under consideration. 

\longshort{\begin{proposition}}{\begin{proposition}[$\star$]} \label{proposition:metodological}
Let $\mathcal{C}$ be a class of structures.  The following are equivalent.
\begin{enumerate}[label=\textit{(\roman*)}]
\item $\textsc{MC}(\mathcal{C},\mathcal{FO}(\exists,\wedge,\vee,\neg))$ is fixed-parameter tractable.
\item $\textsc{MC}(\mathcal{C},\mathcal{FO}(\exists,\wedge,\neg))$ is fixed-parameter tractable.
\item $\textsc{Emb}(\mathcal{C})$ is fixed-parameter tractable.
\end{enumerate}
In particular, $\textsc{Emb}(\mathcal{C})$ polynomial-time (thus fpt) many-one reduces to $\textsc{MC}(\mathcal{C},\mathcal{FO}(\exists,\wedge,\vee,\neg))$.
\end{proposition}

\newcommand{\pfmetodological}[0]{
\begin{proof}
Let $\mathcal{C}$ be a class of 
$\sigma$-structures.  

We give a polynomial-time many-one reduction of 
$\textsc{Emb}(\mathcal{C})$ to $\textsc{MC}(\mathcal{C},\mathcal{FO}(\exists,\wedge,\neg))$.  
Note that embedding a 
$\sigma$-structure $\mathbf{A}$ into a 
$\sigma$-structure $\mathbf{B} \in \mathcal{C}$ 
reduces to checking whether $\mathbf{B}$ verifies the 
existential closure of the $\mathcal{FO}(\wedge,\neg)$-formula  
$$\bigwedge_{a,a' \in A, a \neq a'}a \neq a' 
\wedge 
\bigwedge_{R \in \sigma} 
\left(
\bigwedge_{\mathbf{a} \in R^{\mathbf{A}}} R\mathbf{a}
\wedge 
\bigwedge_{\mathbf{a} \not\in R^{\mathbf{A}}} \neg R\mathbf{a}
\right)\text{.}$$ 

Clearly, $\textsc{MC}(\mathcal{C},\mathcal{FO}(\exists,\wedge,\neg))$ 
polynomial-time many-one reduces to $\textsc{MC}(\mathcal{C},\mathcal{FO}(\exists,\wedge,\vee,\neg))$.  
We conclude the proof giving a fpt Turing (in fact, even truthtable) reduction, 
from $\textsc{MC}(\mathcal{C},\mathcal{FO}(\exists,\wedge,\vee,\neg))$ 
to $\textsc{Emb}(\mathcal{C})$.  

Let $\phi \in \mathcal{FO}(\exists,\wedge,\vee,\neg)$.  Say that $\phi$ 
is \emph{disjunctive} if $\phi=\psi_1 \vee \cdots \vee \psi_l$ 
and $\psi_i \in \mathcal{FO}(\exists,\wedge,\neg)$ for all $i \in [l]$.  
Clearly, for every $\phi \in \mathcal{FO}(\exists,\wedge,\vee,\neg)$, 
a disjunctive $\phi' \in \mathcal{FO}(\exists,\wedge,\vee,\neg)$ 
such that $\phi \equiv \phi'$ is computable by (equivalence 
preserving) syntactic replacements.  

Let $\psi$ be a $\sigma$-sentence in $\mathcal{FO}(\exists,\wedge,\neg)$.  
Say that the disjunctive $\sigma$-sentence $\psi'=\chi_1 \vee \cdots \vee \chi_l$  
is a \emph{completion} of $\psi$ if $\psi' \equiv \psi$ and, 
for all $i \in [l]$, if the quantifier prefix of $\chi_i$ is $\exists x_1 \ldots \exists x_m$, 
then: 
\begin{itemize}
\item 
for all $(y,y') \in \{x_1,\ldots,x_m\}^2$, 
it holds that $y=y'$ or $y \neq y'$ occur in the quantifier free part of $\chi_i$;
\item for all $R \in \sigma$ and all $(y_1,\ldots,y_{\textup{ar}(R)}) \in \{x_1,\ldots,x_m\}^{\textup{ar}(R)}$, 
it holds that $R y_1 \ldots y_{\textup{ar}(R)}$ or $\neg R y_1 \ldots y_{\textup{ar}(R)}$ 
occur in the quantifier free part of $\chi_i$;
\end{itemize}
moreover, $\psi'$ is said \emph{reduced} if, for all $i \in [l]$, $\chi_i$ is satisfiable, 
$\chi_i$ does not contain dummy quantifiers, and $\chi_i$ does not contain atoms of the form $y=y'$.  

Let $\psi'=\chi_1 \vee \cdots \vee \chi_l$ be a reduced completion of 
the $\sigma$-sentence $\psi \in \mathcal{FO}(\exists,\wedge,\neg)$.  
Clearly, $\psi'$ is computable from $\psi$ as follows.  
Let $\exists x_1 \ldots \exists x_{m}$ 
be the quantifier prefix of $\psi$.  
\begin{itemize}
\item For all $(y,y') \in \{x_1,\ldots,x_m\}^2$ such that neither $y=y'$ nor $y \neq y'$ occur in the quantifier free part of $\psi$, 
conjoin $(y=y' \vee y \neq y')$ to the quantifier free part of $\psi$.
\item For all $R \in \sigma$ and $(y_1,\ldots,y_{\textup{ar}(R)}) \in \{x_1,\ldots,x_l\}^{\textup{ar}(R)}$ 
such that neither $R y_1 \ldots y_{\textup{ar}(R)}$ nor\longversion{\\} $\neg R y_1 \ldots y_{\textup{ar}(R)}$ occur in the quantifier free part of $\psi$, 
conjoin $(R y_1 \ldots y_{\textup{ar}(R)} \vee \neg R y_1 \ldots y_{\textup{ar}(R)})$ to the quantifier free part of $\psi$.
\item Compute a disjunctive form of the resulting sentence, eliminate equality atoms and dummy quantifiers from each disjunct, 
and finally eliminate unsatisfiable disjuncts (empty disjunctions are false on all structures).
\end{itemize}

Note that for each $i \in [l]$, the disjunct $\chi_i$ naturally corresponds to 
a $\sigma$-structure $\mathbf{A}_i$, defined as follows.  
Let $\exists x_1 \ldots \exists x_{m}$ be the quantifier prefix of $\chi_i$.  
The universe $A_{\chi_i}$ is $\{x_1,\ldots,x_{m}\}$, 
and $(y_1,\ldots,y_{\textup{ar}(R)}) \in R^{\mathbf{A}_i}$ 
if and only if $R y_1 \ldots y_{\textup{ar}(R)}$ occurs in the quantifier free part of $\chi_i$.

We are now ready to describe the reduction.  Let $(\mathbf{B},\phi)$ be an instance of 
$\textsc{MC}(\mathcal{P},\mathcal{FO}(\exists,\wedge,\vee,\neg))$.  
The algorithm first computes a disjunctive form logically equivalent to $\phi$, 
say $\phi \equiv \psi_1 \vee \cdots \vee \psi_l$, and then, for each $i \in [l]$, 
computes a reduced completion $\psi'_i$ logically equivalent to $\psi_i$, 
say $\psi'_i \equiv \chi'_{i,1} \vee \cdots \vee \chi'_{i,l_i}$.  
For each $i \in [l]$ and $j \in [l_i]$, 
let $\mathbf{A}_{i,j}$ be the structure corresponding to $\chi'_{i,j}$.  

We claim that 
$\mathbf{B} \models \phi$ if and only if 
there exist $i \in [l]$ and $j \in [l_i]$ such that $\mathbf{A}_{i,j}$ embeds into $\mathbf{B}$.  
The backwards direction is clear.  For the forwards direction, 
assume $\mathbf{B} \models \phi$.  Then, there exist $i \in [l]$ and $j \in [l_i]$ 
such that $\mathbf{B} \models \chi'_{i,j}$.  Then, $\mathbf{A}_{i,j}$ embeds 
into $\mathbf{B}$.

Thus, the algorithm works as follows.  For each $i \in [l]$ and $j \in [l_i]$, 
it poses the query $(\mathbf{A}_{i,j},\mathbf{B})$ to the problem $\textsc{Emb}(\mathcal{C})$, 
and it accepts if and only if at least one query answers positively.
\end{proof}
}
\longversion{\pfmetodological}

The next observation is that model checking existential conjunctive logic 
(and thus the full existential logic) on posets is unlikely to be polynomial-time tractable, 
even if the poset is fixed.  Let $\mathbf{B}$ be the bowtie poset 
defined by the universe $B=[4]$ 
and the covers 
$i \prec^\mathbf{B} j$ for all $i \in \{1,2\}$ and $j \in \{3,4\}$.

\longshort{\begin{proposition}}{\begin{proposition}[$\star$]}
\label{pr:exprcomplex}
$\textsc{MC}(\{\mathbf{B}\},\mathcal{FO}(\exists,\wedge,\neg))$ is $\textup{NP}$-hard.
\end{proposition}

\newcommand{\pfexprcomplex}[0]{
\begin{proof}
Let $\sigma=\{\leq,1,2,3,4\}$ be a relational vocabulary where 
$\textup{ar}(\leq)=2$ and $\textup{ar}(i)=1$ for all $i \in [4]$.  
Let $\mathbf{B}^*$ be the $\sigma$-structure 
such that $(B^*,\leq^{\mathbf{B}^*})$ is isomorphic to $\mathbf{B}$, 
say without loss of generality via the isomorphism $f(b)=b \in B^*$ for all $b \in B$, 
and where $b^{\mathbf{B}^*}=\{f(b)\}=\{b\}$ for all $b \in B$.  
By the case $n=2$ of the main theorem in Pratt and Tiuryn \cite[Theorem~2]{PrattTiuryn96}, 
the problem $\textsc{Hom}(\{\mathbf{B}^*\})$ is $\textup{NP}$-hard.  
We give a polynomial-time many-one reduction of $\textsc{Hom}(\{\mathbf{B}^*\})$ 
to $\textsc{MC}(\{\mathbf{B}\},\mathcal{FO}(\exists,\wedge,\neg))$.  

Let $\mathbf{A}$ be an instance of $\textsc{Hom}(\{\mathbf{B}^*\})$, 
and let $\phi$ be the existential closure of the conjunction of the following 
$\{\leq\}$-literals (thus, $\phi$ is a $\mathcal{FO}(\exists,\wedge,\neg)$-sentence 
on the vocabulary of $\mathbf{B}$):
\begin{itemize}
\item $z_i \neq z_j$, for all $1 \leq i<j\leq 4$;
\item $z_i<z_j$, for all $i \in \{1,2\}$ and $j \in \{3,4\}$;
\item $a=z_i$, for all $i \in [4]$ and $a \in i^\mathbf{A}$;
\item $a \leq a'$, for all $a \leq^{\mathbf{A}} a'$.  
\end{itemize}
It is easy to check that $\mathbf{A}$ maps homomorphically to 
$\mathbf{B}^*$ if and only if $\mathbf{B} \models \phi$.
\end{proof}
}
\longversion{\pfexprcomplex}

In contrast, model checking existential logic on any fixed 
poset $\pp$ is trivially fixed-parameter tractable 
(the instance is a structure of constant size, and a sentence taken as a parameter).
However, 
there are classes of 
posets where the embedding problem, 
and hence, by Proposition~\ref{proposition:metodological}, the problem of model checking existential logic, 
is unlikely to be fixed-parameter tractable, as we now show.

First, we introduce a family of poset invariants and relate them as in Figure~\ref{fig:diagram}.  Let $\pp$ be a poset.
\begin{itemize}
\item The \emph{size} of $\pp$ is the cardinality of its universe, $|P|$.   
\item The \emph{width} of $\pp$, in symbols $\textup{width}(\pp)$, 
is the maximum size attained by an antichain in $\pp$.  
\item The \emph{depth} of $\pp$, in symbols $\textup{depth}(\pp)$, 
is the maximum size attained by a chain in $\pp$.  
\item The \emph{degree} of $\pp$, in symbols $\textup{degree}(\pp)$, 
is the degree of the order relation of $\pp$, 
that is, $\textup{degree}(\leq^\pp)$.   
\item The \emph{cover-degree} of $\pp$, 
in symbols $\textup{cover\textup{-}degree}(\pp)$, is the degree of the cover relation of $\pp$, 
that is, $\textup{degree}(\textup{cover}(\pp))$.   
\end{itemize}
\begin{figure}[t]
\centering
\begin{picture}(0,0)%
\includegraphics{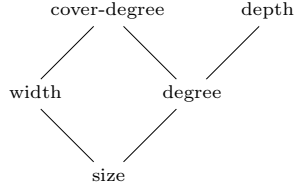}%
\end{picture}%
\setlength{\unitlength}{2279sp}%
\begingroup\makeatletter\ifx\SetFigFont\undefined%
\gdef\SetFigFont#1#2#3#4#5{%
  \reset@font\fontsize{#1}{#2pt}%
  \fontfamily{#3}\fontseries{#4}\fontshape{#5}%
  \selectfont}%
\fi\endgroup%
\begin{picture}(3166,2266)(2468,-10194)
\put(2521,-9106){\makebox(0,0)[lb]{\smash{{\SetFigFont{7}{8.4}{\rmdefault}{\mddefault}{\updefault}{\color[rgb]{0,0,0}$\textup{width}$}%
}}}}
\put(3421,-10006){\makebox(0,0)[lb]{\smash{{\SetFigFont{7}{8.4}{\rmdefault}{\mddefault}{\updefault}{\color[rgb]{0,0,0}$\textup{size}$}%
}}}}
\put(4186,-9106){\makebox(0,0)[lb]{\smash{{\SetFigFont{7}{8.4}{\rmdefault}{\mddefault}{\updefault}{\color[rgb]{0,0,0}$\textup{degree}$}%
}}}}
\put(2971,-8206){\makebox(0,0)[lb]{\smash{{\SetFigFont{7}{8.4}{\rmdefault}{\mddefault}{\updefault}{\color[rgb]{0,0,0}$\textup{cover\textup{-}degree}$}%
}}}}
\put(5041,-8206){\makebox(0,0)[lb]{\smash{{\SetFigFont{7}{8.4}{\rmdefault}{\mddefault}{\updefault}{\color[rgb]{0,0,0}$\textup{depth}$}%
}}}}
\end{picture}%
\caption{The order of poset invariants induced by Proposition~\ref{pr:diagram}.}\label{fig:diagram}
\end{figure}




\longshort{\begin{proposition}}{\begin{proposition}[$\star$]}\label{pr:diagram}
Let $\mathcal{P}$ be a class of 
posets.  
\begin{enumerate}[label=\textit{(\roman*)}]
\item $\mathcal{P}$ has bounded degree if and only if 
$\mathcal{P}$ has bounded depth and bounded cover-degree.  
\item If $\mathcal{P}$ has bounded width, 
then $\mathcal{P}$ has bounded cover-degree.
\item $\mathcal{P}$ has bounded size if and only if 
$\mathcal{P}$ has bounded width and bounded degree.
\end{enumerate}
\end{proposition}

\newcommand{\pfdiagram}[0]{
\begin{proof}
We prove \textit{(i)}.  Assume that $\mathcal{P}$ has bounded degree.  
Let $\pp \in \mathcal{P}$. Then $\textup{cover\textup{-}degree}(\pp)\leq \textup{degree}(\pp)$ 
follows from the fact that $\textup{cover}(\pp)$ is contained in $\leq^\pp$, 
while $\textup{depth}(\pp)\leq \textup{degree}(\pp)+1$ follows from the fact that each chain forms a complete directed acyclic subgraph in $\pp$.
Conversely, 
let $d \in \mathbb{N}$ and $c \in \mathbb{N}$ 
be the largest depth and cover-degree attained by a poset in $\mathcal{P}$, 
respectively.  Then, for every $\pp \in \mathcal{P}$ 
and $p \in P$, it holds that $\textup{degree}(p) \leq c^d$, 
hence $\mathcal{P}$ has bounded degree.

We prove \textit{(ii)}.  Let $w$ be the largest width attained by a poset in $\mathcal{P}$.  
Then, for every $\pp \in \mathcal{P}$ and $p \in P$, 
it holds that $\textup{cover\textup{-}degree}(p) \leq 2w$, 
because the lower covers of $p$ and the upper covers of $p$ 
form antichains in $\pp$, hence $p$ has at most $2w$ lower or upper covers.  
Hence, $\mathcal{P}$ has bounded cover-degree.

We prove \textit{(iii)}.  Assume that $\mathcal{P}$ has bounded size.  
Let $s$ be the largest size attained by a poset in $\mathcal{P}$.  
Then, for every $\pp \in \mathcal{P}$, 
it holds that $\textup{width}(\pp),\textup{degree}(\pp) \leq s$, 
that is, $\mathcal{P}$ has bounded width and bounded degree.  
Conversely, by \textit{(i)}, $\mathcal{P}$ has bounded depth.  
Let $d$ and $w$ be the largest depth and width attained by a poset in $\mathcal{P}$, respectively. 
Let $\pp \in \mathcal{P}$.  By Dilworth's theorem, 
there exist $w$ chains in $\pp$ 
whose union is $P$, hence $\textup{size}(\pp)\leq w \cdot d$.  
We conclude that $\mathcal{P}$ has bounded size.
\end{proof}
}

\longversion{\pfdiagram}

The previous proposition, together with the observation 
that bounded width and bounded degree (bounded width and bounded depth, 
bounded cover-degree and bounded depth, respectively) are incomparable, 
justifies the order in Figure~\ref{fig:diagram}, whose interpretation 
is the following: invariant $\textup{inv}$ is below invariant $\textup{inv}'$ 
if and only if, for every class $\mathcal{P}$ of posets, 
if $\mathcal{P}$ is bounded under $\textup{inv}$, 
then $\mathcal{P}$ is bounded under $\textup{inv}'$.

The emerging hierarchy of poset invariants will provide a measure of tightness 
for our positive algorithmic results, once we will manage to 
surround them with complexity results on covering neighboring classes.  

To this aim, we immediately observe that there exists a class of posets of bounded depth 
where the embedding problem, and hence model-checking existential first-order logic, 
is $\textup{W}[1]$-hard.  Given any graph $\mathbf{G} \in \mathcal{G}$, 
construct a poset $r(\mathbf{G})=\pp$ 
by taking $|G|$ pairwise disjoint $3$-element chains, 
and covering the bottom of the $i$th chain by the top of the $j$th chain 
if and only if $i$ and $j$ are adjacent in $\mathbf{G}$.  
Note that $\textup{depth}(\pp) \leq 3$.  Hence, the class 
$\mathcal{P}_{\textup{depth}}=\{ r(\mathbf{G}) \mid \mathbf{G} \in \mathcal{G} \}$ 
has bounded depth.

\begin{proposition}\label{pr:parhardallposets}
$\textsc{Emb}(\mathcal{P}_\textup{depth})$ is $\textup{W}[1]$-hard.
\end{proposition}

\newcommand{\pfparhardallposets}[0]{
\begin{proof}
$\textsc{Clique}$ fpt many-one reduces to $\textsc{Emb}(\mathcal{P}_{\textup{depth}})$ 
by mapping $(\mathbf{G},k)$ to $(r(\mathbf{K}_k),r(\mathbf{G}))$.  
\end{proof}
}

\longshort{\pfparhardallposets}{\pfparhardallposets}

The goal of the technical part of the paper is to establish 
the facts leading from Figure~\ref{fig:diagram} to Figure~\ref{fig:overwparcompl}:
\begin{itemize}
\item For the parameterized complexity of model checking existential logic, 
we have tractability on bounded degree classes by Seese's algorithm \cite{Seese96}, 
and hardness on (certain) bounded depth classes by Proposition~\ref{pr:parhardallposets}.  
In Section~\ref{sect:mainresults}, we establish tractability on bounded width classes 
by Theorem~\ref{thm:EFOFPT}, and hardness on (certain) bounded cover-degree classes by Proposition~\ref{pr:harddegree}.
\item For the classical complexity of the embedding problem 
 (Section~\ref{sect:classical}), 
Proposition~\ref{th:wdtract} establishes tractability on bounded size classes, 
Theorem~\ref{th:widthnphard} establishes hardness on (certain) bounded width classes, 
and Theorem~\ref{th:degreenphard} establishes hardness on (certain) bounded degree classes.  
\end{itemize}

We conclude the section by relating our work on posets of bounded width 
with previous work on digraphs of bounded clique-width, and showing that 
our results are indeed independent. 

Clique-width is a prominent invariant 
of undirected as well as directed graphs which generalizes treewidth \cite{CourcelleOlariu00}; in particular, 
it is known that monadic second-order logic (precisely, $\mathcal{MSO}_1$) 
is fixed-parameter tractable on digraphs of bounded clique-width \cite{CourcelleMakowskyRotics00}, 
thus:

\begin{observation}
\label{obs:MSO}
$\textsc{MC}(\mathcal{P},\mathcal{FO})$ is fixed-parameter tractable for any class $\mathcal{P}$ of posets such that the clique-width of $\mathcal{P}$ is bounded.
\end{observation}

Since it is possible to compute the cover relation from the order relation (and vice versa) in polynomial time, one might wonder whether using the clique-width of the cover graph would allow us to efficiently model check wider classes of posets. This turns out not to be the case:

\begin{observation}[follows from Examples 1.32, 1.33 and Corollary 1.53 of \cite{CourcelleEngelfriet12}]
\label{obs:cw}
For any class $\mathcal{P}$ of posets, the clique-width of $\mathcal{P}$ is bounded if and only if the clique-width of $\textup{cover}(\mathcal{P})$ is bounded.
\end{observation}

A natural class of posets which is easily observed having clique-width bounded by $2$ (despite having unbounded treewidth) 
is the class of \emph{series parallel posets}.  However, we show that there exist classes of posets of bounded width which 
do not have bounded clique-width (if not Theorem~\ref{thm:EFOFPT} would follow from Observation \ref{obs:MSO}). 

\longshort{\begin{proposition}}{\begin{proposition}[$\star$]}
\label{pr:nocw}
There exists a class $\mathcal{P}$ of posets which has bounded width but does not have bounded clique-width.
\end{proposition}

\newcommand{\pfnocw}[0]{
\begin{proof}
For each $i \in \mathbb{N}$, we define a poset $\mathbf{P}_i$ 
as follows.  The universe is $P_i=\{p_{a,b}, q_{a,b} \mid a,b\in [i] \}$ 
and the cover relation is defined by the following pairs:
\begin{itemize}
\item $p_{a,b}\prec^{\mathbf{P}_i} p_{a,b+1}$ and $q_{a,b}\prec^{\mathbf{P}_i} q_{a,b+1}$,
\item $p_{a,i}\prec^{\mathbf{P}_i} p_{a+1,1}$ and $q_{a,i}\prec^{\mathbf{P}_i} q_{a+1,1}$, 
\item $p_{a,b}\prec^{\mathbf{P}_i} q_{a+1,b}$ and $q_{a,b}\prec^{\mathbf{P}_i} p_{a+1,b}$.
\end{itemize}

Notice that $\textup{cover}(\mathbf{P}_i)$ contains a $i\times i$ grid as a subgraph; indeed, 
one may define the $j$th row of the grid to consist of the chain $p_{1,j}\prec^{\mathbf{P}_i} q_{2,j}\prec^{\mathbf{P}_i} p_{3,j}\prec^{\mathbf{P}_i} q_{4,j}\dots$ 
and similarly the $j$th column to consist of $p_{j,1}\prec^{\mathbf{P}_i} p_{j,2}\prec^{\mathbf{P}_i} p_{j,3}\dots$ for odd $j$ and $q_{j,1}\prec^{\mathbf{P}_i} q_{j,2}\prec^{\mathbf{P}_i} q_{j,3}\dots$ for even $j$. Furthermore, $\mathbf{P}_i$ has width $2$ and $\textup{cover}(\mathbf{P}_i)$ has degree $4$. 
We will prove that $\mathcal{P}=\{\mathbf{P}_i \mid i\in \mathbb{N}\}$ has unbounded clique-width.

Let $\mathcal{H}$ be the class of undirected graphs corresponding to the covers of $\mathcal{P}$ 
(that is, $\mathcal{H}$ contains the symmetric closure of $\textup{cover}(\pp_i)$ 
for all $\pp_i \in \mathcal{P}$).  Since $\mathcal{H}$ contains graphs with arbitrarily large grids, 
$\mathcal{H}$ has unbounded tree-width. Hence $\mathcal{H}$ also has unbounded clique-width by \cite[Corollary 1.53]{CourcelleEngelfriet12}, 
and the fact that it has bounded degree. It is a folklore fact that 
for any graph $\mathbf{G}$ and any orientation $\mathbf{G}'$ of $\mathbf{G}$, 
the clique-width of $\mathbf{G}$ is bounded by the clique-width of $\mathbf{G}'$ (indeed, one can use the same decomposition in this direction).  
Since $\textup{cover}(\mathcal{P})$ contains one orientation for each graph in $\mathcal{H}$ and since $\mathcal{H}$ has unbounded clique-width, we conclude that $\textup{cover}(\mathcal{P})$ has unbounded clique-width.
\end{proof}
}

\longversion{\pfnocw}

\section{Parameterized Complexity}\label{sect:mainresults}
In this section, we study the parameterized complexity of the problems under consideration. 
The section is organized as follows.
\begin{itemize}
\item In Subsection \ref{sect:embfpt}, we develop a fixed-parameter tractable algorithm for the embedding problem on posets of bounded width (Theorem \ref{th:embfpt}), 
which yields that model checking existential logic on such posets is fixed-parameter tractable (Theorem \ref{thm:EFOFPT}).
\item In Subsection \ref{sect:bdcover}, we provide a reduction proving $\textup{W[1]}$-hardness of model checking 
existential logic on posets of bounded cover-degree (Proposition \ref{pr:harddegree}).
\end{itemize}

\subsection{Embedding is FPT on Bounded Width Posets}\label{sect:embfpt}

We first outline our proof strategy.  
The core of the proof lies in defining a suitable 
compilation of bounded width posets. We then proceed in two steps: 
\begin{enumerate}[label=\textit{(\roman*)}] 
\item proving that the homomorphism problem is polynomial-time tractable on such compilations, and
\item reducing the embedding problem between two bounded width posets to fpt many instances 
of the homomorphism problem between compilations of these posets.
\end{enumerate}
For \textit{(i)}, we prove that the compilation admits a semilattice polymorphism (Lemma~\ref{lemma:compres}), 
and use the classical result by Jeavons et al.\ that the homomorphism problem is polynomial-time tractable on semilattice structures (Theorem \ref{th:semilpoly}).  
For \textit{(ii)}, we use color coding and hash functions (Theorem~\ref{th:hash}) to link a homomorphism between two compilations to the existence 
of an embedding between the compiled posets (Lemma \ref{lemma:correct}).


\subsubsection{Known Facts}\label{sect:prim}

The proof uses known facts about semilattice structures and hash functions, collected below.  

\medskip

\noindent \textit{Semilattice Polymorphisms.}  
Let $\sigma$ be a finite relational vocabulary, and let $\mathbf{A}$ be a
$\sigma$-structure.  Let $f \colon A^m \to A$ 
be an $m$-ary function on $A$.  We say that $f$ is a \emph{polymorphism} of $\mathbf{A}$ 
(or, $\mathbf{A}$ \emph{admits} $f$) 
if $f$ \emph{preserves} all relations of $\mathbf{A}$, that is, 
for all $R \in \sigma$, where $\textup{ar}(R)=r$, if 
$$(a_{1,1},a_{1,2},\ldots,a_{1,r}),\ldots,(a_{m,1},a_{m,2},\ldots,a_{m,r}) \in R^\mathbf{A}\text{,}$$  
then 
$$(
f(a_{1,1},a_{2,1},\ldots,a_{m,1}),
\ldots,
f(a_{1,r},a_{2,r},\ldots,a_{m,r})
) \in R^\mathbf{A}\text{.}$$

We say that a function $f \colon A^2 \to A$ is a \emph{semilattice} function over $A$ if 
$f$ is idempotent, associative, and commutative on $A$, 
that is, $f(a,a)=a$, $f(a,f(a',a''))=f(f(a,a'),a'')$, and $f(a,a')=f(a',a)$ for all $a,a',a'' \in A$.  

\begin{theorem}[\cite{JeavonsCohenGyssens97}]\label{th:semilpoly}
Let $\mathbf{A}$ be a $\sigma$-structure, 
and let $f$ be a semilattice function over $A$.  
If $f$ is a polymorphism of $\mathbf{A}$, 
then $\textsc{Hom}(\mathbf{A})$ is polynomial-time tractable.
\end{theorem}

\medskip

\noindent \textit{Hash Functions.}  Let $M$ and $N$ be sets, and let $k \in \mathbb{N}$.  
A \emph{$k$-perfect family of hash functions} from $M$ to $N$ 
is a family $\Lambda$ of functions from $M$ to $N$ 
such that for every subset $K \subseteq M$ of cardinality $k$ 
there exists $\lambda \in \Lambda$ such that $\lambda|_K$ 
is injective. 


\begin{theorem}\label{th:hash}[Theorem 13.14, \cite{FlumGrohe06}]
Let $C$ be a finite set.  There exists an algorithm that, given $C$ and $k \in \mathbb{N}$, 
computes a $k$-perfect family $\Lambda_{C,k}$ of hash functions from $C$ to $[k]$ 
of cardinality $2^{O(k)} \cdot \log^2 |C|$ in time $2^{O(k)} \cdot |C| \cdot \log^2 |C|$.  
\end{theorem}


\newcommand{\exchainpartition}[0]{
\begin{example}\label{ex:chainpartition}
Let $\qq$ be the poset with universe $Q=D_1 \cup D_2$, 
where $D_1=\{ d_{11},d_{12},d_{13},d_{14}\}$ and $D_2=\{ d_{21},d_{22},d_{23},d_{24}\}$, 
and cover relation $d_{11} \prec^{\qq} d_{12} \prec^{\qq} d_{13} \prec^{\qq} d_{14}$, 
$d_{21} \prec^{\qq} d_{22} \prec^{\qq} d_{23} \prec^{\qq} d_{24}$, 
$d_{11} \prec^{\qq} d_{23}$,
$d_{12} \prec^{\qq} d_{24}$, and  
$d_{22} \prec^{\qq} d_{13}$.  Then, 
$(\dd_1,\dd_2)$ is a chain partition of $\qq$.  See Figure~\ref{fig:chainpartition} (left).

Let $\pp$ be the poset with universe $P=C_1 \cup C_2$, 
where $C_1=\{ c_{11},\ldots,c_{16}\}$ and $C_2=\{ c_{21},\ldots,c_{26}\}$, 
and cover relation $c_{11} \prec^{\qq} \cdots \prec^{\qq} c_{16}$,
$c_{21} \prec^{\qq} \cdots \prec^{\qq} c_{26}$, 
$c_{11} \prec^{\qq} c_{24}$, 
$c_{12} \prec^{\qq} c_{25}$, 
$c_{13} \prec^{\qq} c_{26}$, 
$c_{21} \prec^{\qq} c_{14}$, 
$c_{22} \prec^{\qq} c_{15}$, and 
$c_{23} \prec^{\qq} c_{16}$.  Then, 
$(\cc_1,\cc_2)$ is a chain partition of $\pp$.  See Figure~\ref{fig:chainpartition} (right).

The mapping $e \colon Q \to P$ defined by 
$e(d_{11})=c_{11}$,
$e(d_{12})=c_{12}$,
$e(d_{13})=c_{15}$, 
$e(d_{14})=c_{16}$, 
$e(d_{21})=c_{21}$, 
$e(d_{22})=c_{22}$, 
$e(d_{23})=c_{24}$, 
$e(d_{24})=c_{25}$ embeds $\qq$ into $\pp$.
\end{example}

\begin{figure}[h]
\centering
\includegraphics[scale=.2]{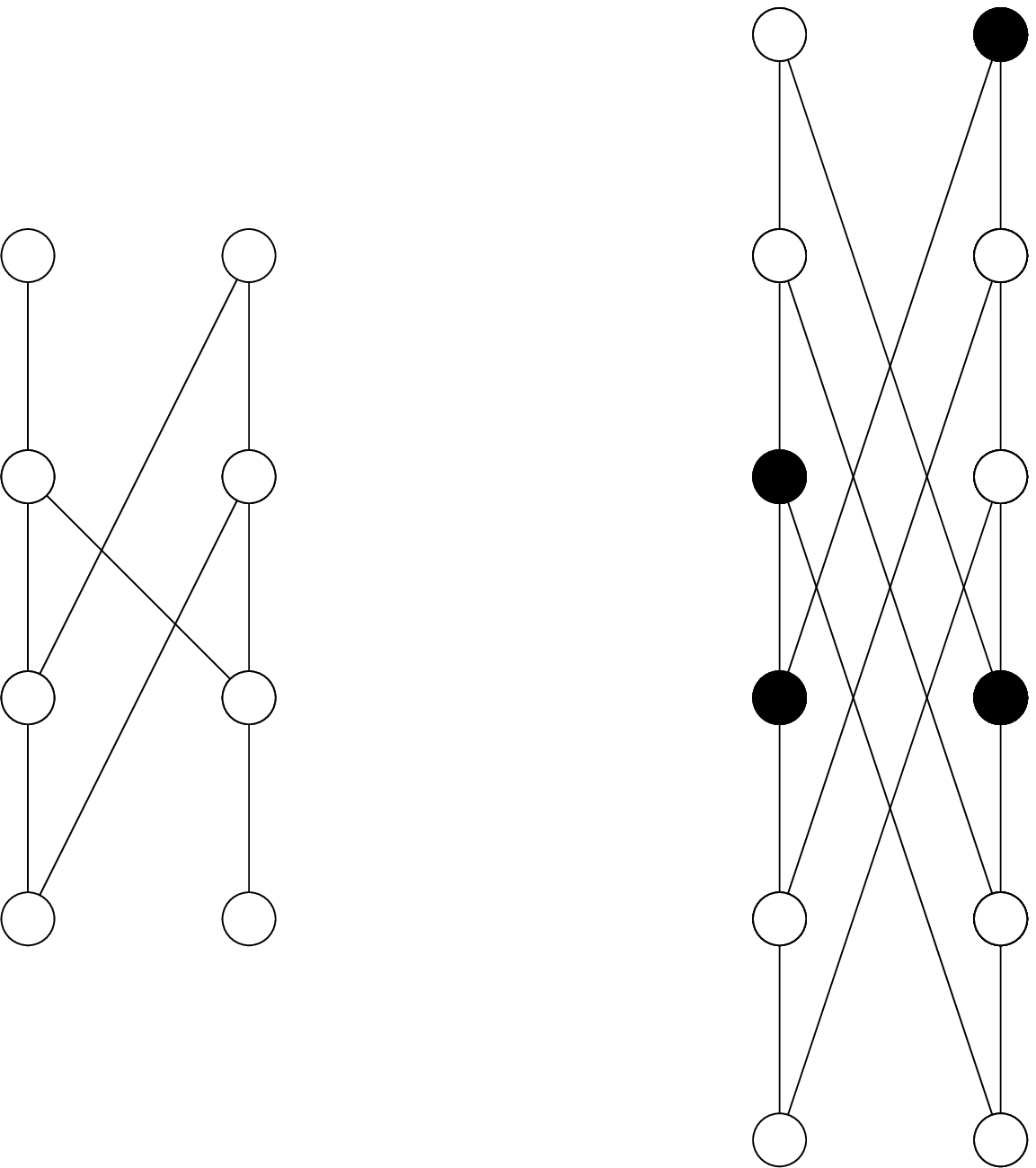}
\caption{The posets $\qq$ (left) and $\pp$ (right) in Example~\ref{ex:chainpartition}.   
The white points in $P$ form the image of the embedding $e \colon Q \to P$ in Example~\ref{ex:chainpartition}.}
\label{fig:chainpartition}
\end{figure}  
}

\subsubsection{Semilattice Compilation}\label{sect:compil}
Let $\pp$ be a 
poset.  Let $(i_1,\ldots,i_a) \in \mathbb{N}^a$ 
be a tuple of numbers.  A \emph{chain partition} of $\pp$ is a tuple $(\cc_{i_1},\ldots,\cc_{i_a})$ 
such that $\emptyset \neq C_{i_j} \subseteq P$ for all $j \in [a]$, 
$P=\bigcup_{j \in [a]}C_{i_j}$, $C_{i_j} \cap C_{i_{j'}}=\emptyset$ for all $1 \leq j<j' \leq a$, 
$\cc_{i_j}$ is the substructure of $\pp$ induced by $C_{i_j}$, 
and $\cc_{i_j}$ is a chain.  \longshort{\exchainpartition}{\exchainpartition}


\begin{theorem}\label{th:felsner}[Theorem 1, \cite{FelsnerRaghavanSpinrad03}]
Let $\pp$ be a 
poset.  Then, 
in time $O(\textup{width}(\pp) \cdot |P|^2)$, 
it is possible to compute both $\textup{width}(\pp)$ 
and a chain partition of $\pp$ of the form $(\cc_1,\ldots,\cc_{\textup{width}(\pp)})$.
\end{theorem}

\newcommand{\excompilp}[0]{
\begin{example}\label{ex:compilp}
Let $\qq$ and $(\dd_1,\dd_2)$ be as in Example~\ref{ex:chainpartition}.  
Let the subtuple of $(1,2)$ be $(1,2)$ itself.  Let $k_1=k_2=4=|D_1|=|D_2|$.  
Let $\mu_1 \colon D_1 \to [k_1]$ be defined by 
$\mu_1(d_{11})=1$, 
$\mu_1(d_{12})=2$, 
$\mu_1(d_{13})=3$, and  
$\mu_1(d_{14})=4$.  
Let $\mu_2 \colon D_2 \to [k_2]$ be defined by 
$\mu_2(c_{21})=1$, 
$\mu_2(c_{22})=2$, 
$\mu_2(c_{23})=3$, and  
$\mu_2(c_{24})=4$.  Then, 
$\textup{compil}(\qq,\dd_{1},\dd_{2},\mu_1,\mu_{2})$ 
is depicted in Figure~\ref{fig:compilq}.

Let $\pp$ and $(\cc_1,\cc_2)$ be as in Example~\ref{ex:chainpartition}.  
Let the subtuple of $(1,2)$ be $(1,2)$ itself.  Let $k_1=k_2=4 \leq 6=|C_1|=|C_2|$.  
Let $\lambda_1 \colon C_1 \to [k_1]$ be defined by 
$\lambda_1(c_{11})=1$, 
$\lambda_1(c_{12})=2$, 
$\lambda_1(c_{13})=4$, 
$\lambda_1(c_{14})=1$, 
$\lambda_1(c_{15})=3$, and 
$\lambda_1(c_{16})=4$.  Let $\lambda_2 \colon C_2 \to [k_2]$ be defined by 
$\lambda_2(c_{21})=1$, 
$\lambda_2(c_{22})=2$, 
$\lambda_2(c_{23})=3$, 
$\lambda_2(c_{24})=3$, 
$\lambda_2(c_{25})=4$, and 
$\lambda_2(c_{26})=1$.  Then, 
$\textup{compil}(\pp,\cc_{1},\cc_{2},\lambda_1,\lambda_{2})$ 
is depicted in Figure~\ref{fig:compilp}.  

\begin{figure}[h]
\centering
\includegraphics[scale=.19]{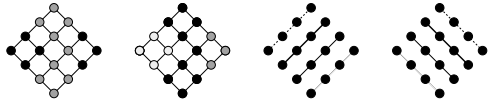}
\caption{Describing the structure $\textup{compil}(\qq,\dd_{1},\dd_{2},\mu_1,\mu_{2})$ in Example~\ref{ex:compilp}.  
From left to right.  
The first picture displays the interpretation of $L$ (thin solid edges) 
and $I_{\{1,2\}}$ (gray points) induced by \textit{(i)} and \textit{(ii)}.  
The second picture displays the interpretation of $L$ (thin solid edges), 
$O_{(2,1)}$ (light gray points), and $O_{(1,2)}$ (dark gray points) 
induced by \textit{(i)} and \textit{(ii)}.  
The third picture displays the interpretation of $R_{(1,1)}$ (dotted edges), 
$R_{(1,2)}$ (medium solid edges), $R_{(1,3)}$ (thick solid edges), and $R_{(1,4)}$ (dashed edges), 
as induced by \textit{(iii)} and $\lambda_1$.  Similarly, 
the fourth picture displays the interpretation of 
$R_{(2,1)}$, $R_{(2,2)}$, $R_{(2,3)}$, and $R_{(2,4)}$ 
induced by \textit{(iii)} and $\lambda_2$.}
\label{fig:compilq}
\end{figure}  

\begin{figure}[h]
\centering
\includegraphics[scale=.19]{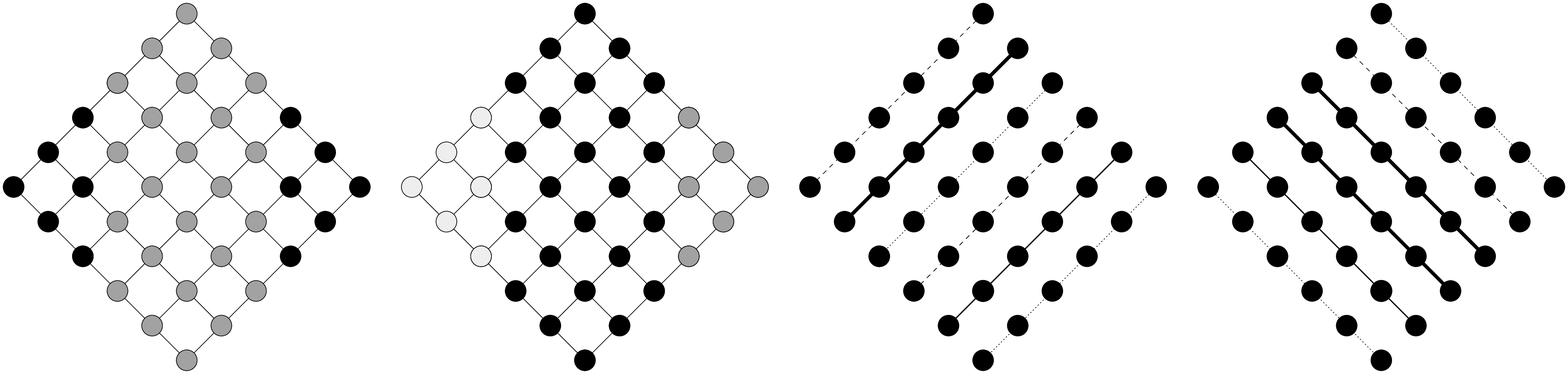}
\caption{Describing the structure $\textup{compil}(\pp,\cc_{1},\cc_{2},\lambda_1,\lambda_{2})$ in Example~\ref{ex:compilp}, 
along the lines of Figure~\ref{fig:compilq}.}
\label{fig:compilp}
\end{figure}  
\end{example}}

We are now ready to define the aforementioned compilations.  
Note that our compilations will depend not only on the poset itself, 
but also on a chain decomposition of the poset and a family of colorings 
(the significance of the latter will become clear in the proof of Lemma \ref{lemma:correct}).  


Let $\pp$ be a 
poset such that $\textup{width}(\pp)\leq w$, 
and let $(\cc_1,\ldots,\cc_{w})$ be a chain partition of $\pp$.  
Let $w' \leq w$ and let $(i_1,\ldots,i_{w'})$ be a subtuple of $(1,\ldots,w)$, 
that is, $(i_1,\ldots,i_{w'})$ is obtained from $(1,\ldots,w)$ by deleting $w-w'$ indices.  
For all $j \in [w']$, let $k_{i_j} \in \mathbb{N}$ be such that $k_{i_j} \leq |C_{i_j}|$, 
$\Lambda_j$ be a family of functions from $C_{i_j}$ to $[k_{i_j}]$, 
and $(\lambda_1,\ldots,\lambda_{w'}) \in \Lambda_1 \times \cdots \times \Lambda_{w'}$.  

For a suitable relational vocabulary $\sigma$ depending on $w'$ and $k_{i_j}$ for all $j \in [w']$, 
we define the $\sigma$-structure $$\textup{compil}(\pp,\cc_{i_1},\ldots,\cc_{i_{w'}},\lambda_1,\ldots,\lambda_{w'})\text{,}$$
which we call the \emph{compilation} of $\pp$ with respect to the \emph{coordinatization} $(\cc_{i_1},\ldots,\cc_{i_{w'}})$ 
and the \emph{coloring} $(\lambda_1,\ldots,\lambda_{w'})$, as follows 
(we use $\textup{compil}(\pp)$ as a shorthand
if the coordinatization and the coloring are contextually clear).


The relational vocabulary $\sigma$ of $\textup{compil}(\mathbf{P})$ 
consists of one binary relation symbol $L$, 
two unary relation symbols $I_{\{j,j'\}}$ and $O_{(j,j')}$ 
for each $2$-element subset $\{j,j'\}$ of $[w']$, 
and one binary relation symbol $R_{(j,k)}$ 
for each $j \in [w']$ and $k \in [k_{i_j}]$.  

The universe of $\textup{compil}(\mathbf{P})$ is 
$$\textup{compil}(P)=C_{i_1} \times C_{i_2} \times \cdots \times C_{i_{w'}}\text{.}$$

Let $\mathbf{c}=(c_1,\ldots,c_{w'})$ and $\mathbf{c}'=(c'_1,\ldots,c'_{w'})$ 
be elements of $\textup{compil}(\mathbf{P})$, and let $K_{(j,k)}=\{ c \in C_{i_j} \mid \lambda_j(c)=k \}$.  
The interpretation of the vocabulary $\sigma$ 
in $\textup{compil}(\mathbf{P})$ is the following:
\begin{enumerate}[label=\textit{(\roman*)}] 
\item The interpretation of $L$ is the set of all pairs $(\mathbf{c},\mathbf{c}')$ 
such that $c_1 \leq^\mathbf{P} c'_1,\ldots,c_{w'} \leq^\mathbf{P} c'_{w'}$.
\item For each $2$-element subset $\{j,j'\}$ of $[w']$, 
$I_{\{j,j'\}}$ and $O_{(j,j')}$ are interpreted, respectively, over  
$I_{\{j,j'\}}=\{ \mathbf{c} \mid c_{j} \parallel^\mathbf{P} c_{j'} \}$  
and $O_{(j,j')}=\{ \mathbf{c} \mid c_{j} <^\mathbf{P} c_{j'} \}$,
\item For each $j \in [w']$ and $k \in [k_{i_j}]$, 
$R_{(j,k)}$ is interpreted over the subset of the interpretation of $L$ defined by 
$$\{ (\mathbf{c},\mathbf{c}') \in L^{\textup{compil}(\mathbf{P})} \suchthat 
\text{$c_{j}\in K_{(j,k)},c_{j}=c'_{j}$} \} 
\text{.}$$
\end{enumerate}
%
%

\longshort{\excompilp}{\excompilp}

The intuition underlying the compilation procedure is the following.  
The universe of $\textup{compil}(\pp)$ is the Cartesian product 
of a family of chains $\cc_{i_1},\ldots,\cc_{i_{w'}}$ partitioning the universe of $\pp$.  
The interpretation of $L$ in $\textup{compil}(\pp)$ is the natural lattice order inherited 
by $\textup{compil}(\pp)$ from $\cc_{i_1},\ldots,\cc_{i_{w'}}$.  For 
$\{i,j\} \subseteq [w']$, the interpretations of 
$I_{\{j,j'\}}$ and $O_{(j,j')}$ in $\textup{compil}(\pp)$ 
record, respectively, incomparabilities and comparabilities between the $j$th and $j'$th 
coordinate (corresponding to elements in the chains $\cc_{i_j}$ and $\cc_{i_{j'}}$, respectively) 
of the tuples in $\textup{compil}(P)$.  Finally, 
for each $j \in [w']$ and $k \in [k_{i_j}]$, 
the interpretation of $R_{(j,k)}$ in $\textup{compil}(\pp)$ 
is the restriction of the lattice order of $\textup{compil}(\pp)$ 
to those pairs of tuples in $\textup{compil}(P)$ such that their 
$j$th coordinate is colored $k$ by $\lambda_j$; the $R_{(j,k)}$'s 
responsibility is to implement the color coding technique (in our setting), 
as it will become clear in the proof of Claim~\ref{cl:cl1}.  

We define a binary function $$s \colon \textup{compil}(P)^2 \to \textup{compil}(P)$$ 
as follows.  Let $\mathbf{c}=(c_1,\ldots,c_{w'})$ and $\mathbf{c}'=(c'_1,\ldots,c'_{w'})$ 
be elements in $\textup{compil}(P)$.  Let $j \in [w']$.  
Recalling that $\cc_{i_j}$ is a chain, 
let $d_j=\textup{min}^{\cc_{i_j}}(c_j,c'_j)$.  Define 
\begin{equation}\label{eq:semilattice}
s(\mathbf{c},\mathbf{c}')=(d_1,\ldots,d_{w'})\text{.} 
\end{equation}

Clearly, $s$ is idempotent, associative and commutative, and hence $s$ is a semilattice function over $\textup{compil}(P)$.  

\longshort{\begin{lemma}}{\begin{lemma}[$\star$]}
\label{lemma:compres}
Let $\pp$ be a 
poset, 
$(\cc_{i_1},\ldots,\cc_{i_{w'}})$ be a coordinatization of $\pp$, 
$(\lambda_1,\ldots,\lambda_{w'})$ be a coloring of $\pp$.  Then, 
the function $s$ in (\ref{eq:semilattice}) is a polymorphism of $\textup{compil}(\pp,\cc_{i_1},\ldots,\cc_{i_{w'}},\lambda_1,\ldots,\lambda_{w'})$.
\end{lemma}

\newcommand{\pfcompres}[0]{
\begin{proof}
We denote $\textup{compil}(\pp,\cc_{i_1},\ldots,\cc_{i_{w'}},\lambda_1,\ldots,\lambda_{w'})$ 
by $\textup{compil}(\pp)$ in short.  We check that $s$ preserves each relation in the vocabulary.  
In the rest of the proof, $\mathbf{c}=(c_1,\ldots,c_{w'})$, 
$\mathbf{c}'=(c'_1,\ldots,c'_{w'})$, 
$\mathbf{d}=(d_1,\ldots,d_{w'})$, and 
$\mathbf{d}'=(d'_1,\ldots,d'_{w'})$ are elements of $\textup{compil}(\pp)$. 

$L$ in $\sigma$: We claim that $s$ preserves $L$.  
Let 
$(\mathbf{c},\mathbf{c}'),(\mathbf{d},\mathbf{d}') \in L$.  
Suffices to show that, for all $j \in [w']$, 
$\textup{min}^{\pp}(c_j,d_j) \leq^{\pp} \textup{min}^{\pp}(c'_j,d'_j)$.  

By hypothesis\longshort{ we have}{,} 
$c_1 \leq^\pp c'_1,\ldots,c_{w'} \leq^\pp c'_{w'}$
and
$d_1 \leq^\pp d'_1,\ldots,d_{w'} \leq^\pp d'_{w'}$ 
so that 
$c_1 \leq^{\cc_{i_1}} c'_1,\ldots,c_{w'} \leq^{\cc_{i_{w'}}} c'_{w'}$
and
$d_1 \leq^{\cc_{i_1}} d'_1,\ldots,d_{w'} \leq^{\cc_{i_{w'}}} d'_{w'}$.  
For all $j \in [w']$, 
$c_j \leq^{\cc_{i_j}} c'_j$ and $d_j \leq^{\cc_{i_j}} d'_j$ 
implies 
$\textup{min}^{\cc_{i_j}}(c_j,d_j) \leq^{\cc_{i_j}} \textup{min}^{\cc_{i_j}}(c'_j,d'_j)$, 
which implies $\textup{min}^{\pp}(c_j,d_j) \leq^{\pp} \textup{min}^{\pp}(c'_j,d'_j)$, 
and we are done.

$I_{\{j,j'\}}$ in $\sigma$ for $1 \leq j < j' \leq w'$: We claim that $s$ preserves $I_{\{j,j'\}}$.  
Let $\mathbf{c},\mathbf{d} \in I_{\{j,j'\}}$.  
Suffices to show that 
$\textup{min}^{\cc_{i_j}}(c_j,d_j) \parallel^\pp \textup{min}^{\cc_{i_{j'}}}(c_{j'},d_{j'})$.  

Assume the contrary for a contradiction, say 
$\textup{min}^{\cc_{i_j}}(c_j,d_j) \leq^\pp \textup{min}^{\cc_{i_{j'}}}(c_{j'},d_{j'})$ 
(the other case is similar).  If $c_j \leq^{\cc_{i_j}} d_j$ and $c_{j'} \leq^{\cc_{i_{j'}}} d_{j'}$, 
then $c_j \leq^\pp c_{j'}$, contradicting the hypothesis that $c_j \parallel^\pp c_{j'}$.  
Similarly, it is impossible that $d_j \leq^{\cc_{i_j}} c_j$ and $d_{j'} \leq^{\cc_{i_{j'}}} c_{j'}$.  
So, assume that $c_j \leq^{\cc_{i_j}} d_j$ and $d_{j'} \leq^{\cc_{i_{j'}}} c_{j'}$.  
Then, $c_j \leq^\pp d_{j'} \leq^\pp c_{j'}$ by the absurdum hypothesis and the case distinction, 
a contradiction.  The case $d_j \leq^{\cc_{i_j}} c_j$ and $c_{j'} \leq^{\cc_{i_{j'}}} d_{j'}$ is similar.

$O_{(j,j')}$ and $O_{(j',j)}$ in $\sigma$ for $1 \leq j < j' \leq w'$: 
We claim that $s$ preserves $O_{(j,j')}$ 
and $O_{(j',j)}$.  We argue for $O_{(j,j')}$, 
and $O_{(j',j)}$ is similar.  Let $\mathbf{c},\mathbf{d} \in O_{(j,j')}$.  
Suffices to show that $\textup{min}^{\cc_{i_j}}(c_j,d_j) \leq^{\pp} \textup{min}^{\cc_{i_{j'}}}(c_{j'},d_{j'})$, 
since $\cc_{i_j} \cap \cc_{i_{j'}}=\emptyset$.  

If $c_{j} \leq^{\cc_{i_j}} d_{j}$ and $c_{j'} \leq^{\cc_{i_{j'}}} d_{j'}$, 
then $c_{j} \leq^\pp c_{j'}$ by hypothesis; similarly 
if $d_{j} \leq^{\cc_{i_j}} c_{j}$ and $d_{j'} \leq^{\cc_{i_{j'}}} c_{j'}$.  
So, assume that $c_{j} \leq^{\cc_{i_j}} d_{j}$ and $d_{j'} \leq^{\cc_{i_{j'}}} c_{j'}$.  
Combining the main hypothesis and the case distinction, 
we have $c_{j} \leq^{\cc_{i_j}} d_{j} \leq^{\pp} d_{j'}$, 
that is, $c_{j} \leq^{\pp} d_{j'}$.  
Similarly, $d_{j} \leq^{\cc_{i_j}} c_{j}$ and $c_{j'} \leq^{\cc_{i_{j'}}} d_{j'}$ 
implies $d_{j} \leq^{\pp} c_{j'}$.  

$R_{(j,k)}$ for $j \in [w']$ and $k \in [k_{i_j}]$: 
We claim that $s$ preserves $R_{(j,k)}$.  

To prove the claim, let 
$(\mathbf{c},\mathbf{d}),(\mathbf{c}',\mathbf{d}') \in R$. 
Let $b,b' \in K_{(j,k)}$ be such that $c_{j}=d_{j}=b$ and $c'_{j}=d'_{j}=b'$.  
Assume $b \leq^{\cc_{i_j}} b'$ (the other case is similar).  Clearly, 
$\textup{min}^{\cc_{i_j}}(c_{j},c'_{j})=\textup{min}^{\cc_{i_{j}}}(d_{j},d'_{j})=b$.  
By hypothesis, $(\mathbf{c},\mathbf{d}),(\mathbf{c}',\mathbf{d}') \in L$, 
so that, by the above, $$(s(\mathbf{c},\mathbf{c}'), 
s(\mathbf{d},\mathbf{d}')) \in L\text{,}$$ 
and thus, by definition, 
$$(s(\mathbf{c},\mathbf{c}'), 
s(\mathbf{d},\mathbf{d}')) \in R\text{,}$$ 
which completes the proof.
\end{proof} 
}

\longversion{\pfcompres}

It follows from Lemma \ref{lemma:compres} and Theorem \ref{th:semilpoly} 
that, for every poset $\mathbf{P}$ and every compilation $\mathbf{P}^*$ of $\mathbf{P}$, 
the problem $\textsc{Hom}(\mathbf{P}^*)$ is polynomial-time tractable; 
this settles the main result of this section.


\subsubsection{Reduction}\label{sect:reduction}
The following lemma reduces an instance of the poset embedding problem 
to a family of instances of the homomorphism problem for suitable compilations of the given posets.  
The lemma is illustrated in Example \ref{ex:compilhom}.

\begin{lemma}
\label{lemma:correct}
Let $\qq$ and $\pp$ be 
posets such that $\textup{width}(\qq) \leq \textup{width}(\pp)=w$. 
Let $(\cc_1,\ldots,\cc_{w})$ be a chain partition of $\pp$.  The following are equivalent.
\begin{enumerate}[label=\textit{(\roman*)}] 
\item $\qq$ embeds into $\pp$. 
\item There exist $w' \leq w$, 
a subtuple $(i_1,\ldots,i_{w'})$ of $(1,\ldots,w)$, 
a chain partition $(\dd_{i_1},\ldots,\dd_{i_{w'}})$ of $\qq$ 
such that $|D_{i_j}| \leq |C_{i_j}|$ for all $j \in [w']$, 
and a tuple $(\mu_1,\ldots,\mu_{w'})$ of bijections from $D_{i_j}$ to $[|D_{i_j}|]$ for all $j \in [w']$, 
such that, for all tuples $(\Lambda_{1},\ldots,\Lambda_{w'})$, 
where $\Lambda_j$ is a $|D_{i_j}|$-perfect family of hash functions from 
$C_{i_j}$ to $[|D_{i_j}|]$ for all $j \in [w']$, 
there exists a tuple $(\lambda_1,\ldots,\lambda_{w'}) \in \Lambda_{1} \times \ldots \times \Lambda_{w'}$ such that 
such that $$\mathbf{Q}^* \in \textsc{Hom}(\mathbf{P}^*)\text{,}$$
where 
\begin{align*}
\mathbf{Q}^*&=\textup{compil}(\mathbf{Q},\dd_{i_1},\ldots,\dd_{i_{w'}},\mu_1,\ldots,\mu_{w'})\text{,}\\
\mathbf{P}^*&=\textup{compil}(\mathbf{P},\cc_{i_1},\ldots,\cc_{i_{w'}},\lambda_1,\ldots,\lambda_{w'})\text{.} 
\end{align*}
\end{enumerate}
\end{lemma}


\newcommand{\pfforwclaima}[0]{
\begin{proof}\renewcommand{\qedsymbol}{$\dashv$}[Proof of Claim~\ref{claim:forwclaim1}]
To prove the claim, let $e(Q)=\{ e(q) \mid q \in Q \}$.  Let $(i_1,i_2,\ldots,i_{w'})$ be the subtuple of $(1,2,\ldots,w)$ 
uniquely determined by deleting the index $i \in [w]$ if and only if $e(Q) \cap C_{i}=\emptyset$.  
For all $j \in [w']$, let $D_{i_j}=e^{-1}(C_{i_j})$, 
and let $\mathbf{D}_{i_j}$ be the substructure of $\mathbf{Q}$ induced by $D_{i_j}$.  
Then, $(\mathbf{D}_{i_1},\ldots,\mathbf{D}_{i_{w'}})$ is a chain partition of $\mathbf{Q}$, 
and clearly $e(D_{i_j}) \subseteq C_{i_j}$ for all $j \in [w']$, which settles the claim.
\end{proof}}

\newcommand{\pfforwclaim}[0]{
\begin{proof}\renewcommand{\qedsymbol}{$\dashv$}[Proof of Claim~\ref{claim:forwclaim}]
Note that $\qq^*$ and $\pp^*$ have the same vocabulary.  
To prove the claim, we check that $h$ preserves all relations in the vocabulary.  
Below, $\mathbf{d}=(d_1,\ldots,d_{w'})$ and $\mathbf{d}'=(d'_1,\ldots,d'_{w'})$ are elements of $\qq^*$.

$L$:  If $(\mathbf{d},\mathbf{d}') \in L^{\qq^*}$, 
then $d_j \leq^{\dd_{i_j}} d'_j$ for all $j \in [w']$, 
then $d_j \leq^{\qq} d'_j$ for all $j \in [w']$, 
then $e(d_j) \leq^{\pp} e(d'_j)$ for all $j \in [w']$, 
then $e(d_j) \leq^{\cc_{i_j}} e(d'_j)$ for all $j \in [w']$, 
then $((e(d_1),\ldots,e(d_{w'})),(e(d'_1),\ldots,e(d'_{w'}))) \in L^{\pp^*}$. 
Altogether this yields $(h(\mathbf{d}),h(\mathbf{d}')) \in L^{\pp^*}$.

$I_{\{j,j'\}}$:  If $\mathbf{d} \in I_{\{j,j'\}}^{\qq^*}$, 
then $d_j \parallel^{\qq} d_{j'}$, 
then $e(d_j) \parallel^{\pp} e(d_{j'})$, 
then $(e(d_1),\ldots,e(d_{w'})) \in I_{\{j,j'\}}^{\pp^*}$, 
that is, $h(\mathbf{d}) \in I_{\{j,j'\}}^{\pp^*}$.

$1 \leq j<j' \leq w'$, $O_{(j,j')}$: If $\mathbf{d} \in O_{(j,j')}^{\qq^*}$, 
then $d_j <^{\qq} d_{j'}$, 
then $e(d_j) <^{\pp} e(d_{j'})$, 
then $(e(d_1),\ldots,e(d_{w'})) \in O_{(j,j')}^{\pp^*}$, 
that is, $h(\mathbf{d}) \in O_{(j,j')}^{\pp^*}$.  The case $O_{(j',j)}$ is similar.

$j \in [w']$, $k \in [|D_{i_j}|]$, $R_{(j,k)}$:   
If $(\mathbf{d},\mathbf{d}') \in R_{(j,k)}^{\qq^*}$, 
then first observe that $(\mathbf{d},\mathbf{d}') \in L^{\qq^*}$, 
so that $(h(\mathbf{d}),h(\mathbf{d}')) \in L^{\pp^*}$ by the argument above.  
We have that $d_{j}=d'_{j}=d$ for some $d \in D_{i_j}$ such that $\mu_{j}(d)=k$.  
By construction, $\mu_{j}(d)=k$ if and only if there exists $c \in \cc_{i_j}$ such that 
$d=e^{-1}(c)$ 
and $\lambda_{j}(c)=k$.   Therefore, 
$e(d_{j})=e(d'_{j})=e(d)=c$, so that 
$((e(d_1),\ldots,e(d_{w'})),(e(d'_1),\ldots,e(d'_{w'}))) \in R_{(j,k)}^{\pp^*}$, 
that is, 
$(h(\mathbf{d}),h(\mathbf{d}')) \in R_{(j,k)}^{\pp^*}$.  
\end{proof}}

\newcommand{\pfcla}[0]{
\begin{proof}\renewcommand{\qedsymbol}{$\dashv$}[Proof of Claim~\ref{cl:cl1}]
Let $\mu_j(q)=k$.  Since $\{ \mathbf{d} \in Q^* \mid d_j=q \}$ is nonempty, 
there exists at least one element $p \in C_{i_j}$ such that 
$\mathbf{c}$ is in the image of $h$ in $P^*$ and $c_j=p$.  
Let $p,p' \in C_{i_j}$ be such that, 
for some $\mathbf{d},\mathbf{d}' \in Q^*$ with $d_j=d'_j=q$, 
$h(\mathbf{d})=\mathbf{c}$ and $c_j=p$, 
and $h(\mathbf{d}')=\mathbf{c}'$ and $c'_j=p'$.  
We prove that $p=p'$ and $\lambda_j(p)=k$.    We distinguish two cases.    

Case $1$:  $(\mathbf{d},\mathbf{d}') \in L^{\qq^*}$ 
or $(\mathbf{d}',\mathbf{d}) \in L^{\qq^*}$.  
Assume $(\mathbf{d},\mathbf{d}') \in L^{\qq^*}$.  
Then, $(\mathbf{d},\mathbf{d}') \in R_{j,k}^{\qq^*}$.  
Then, $(\mathbf{c},\mathbf{c}') \in R_{j,k}^{\pp^*}$, 
so that $c_j=c'_j$ by definition of $R_{j,k}^{\pp^*}$, that is $p=p'$, and $\lambda_j(p)=k$.  
The argument is similar if $(\mathbf{d}',\mathbf{d}) \in L^{\qq^*}$.

Case $2$:  $(\mathbf{d},\mathbf{d}') \not\in L^{\qq^*}$ 
and $(\mathbf{d}',\mathbf{d}) \not\in L^{\qq^*}$.  
Clearly it then holds that $\textup{min}^{\qq_{i_j}}(d_{j},d'_{j})=q$ and 
$\textup{min}^{\qq_{i_{j'}}}(d_{j'},d'_{j'}) \leq^{\qq_{i_{j'}}} d_{j'},d'_{j'}$ for all $j' \in [w']$.  
Therefore, 
$$(
(\textup{min}^{\qq_{i_1}}(d_1,d'_1),\ldots,\textup{min}^{\qq_{i_{w'}}}(d_{w'},d'_{w'})),
\mathbf{d}
) \in R_{j,k}^{\qq^*}\text{,}$$ 
and 
$$(
(\textup{min}^{\qq_{i_1}}(d_1,d'_1),\ldots,\textup{min}^{\qq_{i_{w'}}}(d_{w'},d'_{w'})),
\mathbf{d}'
) \in R_{j,k}^{\qq^*}\text{.}$$
Let $$h((\textup{min}^{\qq_{i_1}}(d_1,d'_1),\ldots,\textup{min}^{\qq_{i_{w'}}}(d_{w'},d'_{w'})))=\mathbf{c}''\text{.}$$
Then, 
$(\mathbf{c}'',\mathbf{c}) \in R_{j,k}^{\pp^*}$ 
and 
$(\mathbf{c}'',\mathbf{c}') \in R_{j,k}^{\pp^*}$, 
so that $c''_j=c_j=c'_j$, that is $p=p'$, and $\lambda_j(p)=k$. 
\end{proof}}

\newcommand{\pfclb}[0]{
\begin{proof}\renewcommand{\qedsymbol}{$\dashv$}[Proof of Claim~\ref{cl:cl2}]
Let $q,q' \in Q$.  
It is sufficient to check that $q<^\qq q'$ implies $e(q)<^\pp e(q')$, 
and $q\parallel^\qq q'$ implies $e(q)\parallel^\pp e(q')$.  
Let $j,j' \in [w']$ be such that $q \in D_{i_j}$ and $q' \in D_{i_{j'}}$.  

$q<^\qq q'$ implies $e(q)<^\pp e(q')$:  Assume $q<^\qq q'$.  Assume that $j \leq j'$ (the case $j' \leq j$ is similar).  
We distinguish two cases.    

Case $1$:  If $j=j'$, then let $\mu_j(q)=k$ and $\mu_j(q')=k'$.    Since $q,q' \in D_{i_j}$ and $\mu_j$ is a bijection from $D_{i_j}$ 
to $[|D_{i_j}|]$, we have that $k \neq k'$.  Hence, if $e(q)=p \in C_{i_j}$ and $e(q')=p' \in C_{i_j}$, 
then by the definition of $e$ we have that $\lambda_j(p)=k\neq k'=\lambda_j(p')$, so that $p \neq p'$.  
We have that 
\begin{align*}
( &(\textup{bot}(\dd_{i_1}),\ldots,q,\ldots,\textup{bot}(\dd_{i_{w'}})), \\
  &(\textup{bot}(\dd_{i_1}),\ldots,q',\ldots,\textup{bot}(\dd_{i_{w'}}))) 
\in L^{\qq^*}\text{,} 
\end{align*}
where $q$ and $q'$ occur at the $j$th coordinate, 
and $\textup{bot}(\dd_{i_{j''}})$ is the bottom of chain $\dd_{i_{j''}}$ for all $j'' \in [w']\setminus\{j\}$. Let\shortversion{ us set}
$h((\textup{bot}(\dd_{i_1}),\ldots,q,\ldots,\textup{bot}(\dd_{i_{w'}})))=\mathbf{c} \in P^*$ 
and similarly 
$h((\textup{bot}(\dd_{i_1}),\ldots,q',\ldots,\textup{bot}(\dd_{i_{w'}})))=\mathbf{c}' \in P^*$.
Then $$(\mathbf{c},\mathbf{c}') \in L^{\pp^*}\text{,}$$
so that, in particular, $c_j \leq^{\cc_{i_j}} c'_j$.  
We claim that $c_j=p$.  Indeed, since $h$ is a homomorphism, 
it is the case that $\mu_j(q)=\lambda_j(c_j)=k$, because 
there is a $R_{(j,k)}$ loop over the elements of $(\textup{bot}(\dd_{i_1}),\ldots,q,\ldots,\textup{bot}(\dd_{i_{w'}}))$ 
in $\qq^*$.  
By Claim~\ref{cl:cl1}, there exists a unique element in $\cc_{i_j}$ 
having the same color of $q$ and occurring at the $j$th coordinate 
of any $h((\ldots,q,\ldots)) \in P^*$, and this element is $e(q)=p$ by definition.  
Similarly, $c'_j=p'$.  Thus, since we observed that $p \neq p'$, 
we have that $p <^{\cc_{i_j}} p'$, and therefore, $e(q)=p <^{\pp} p'=e(q')$.  

Case $2$:  If $j<j'$, then $e(q)=p \in C_{i_j}$ and $e(q')=p' \in C_{i_{j'}}$, so that $p \neq p'$ 
because $C_{i_j} \cap C_{i_{j'}}=\emptyset$.  We have that 
$$(\ldots,q,\ldots,q',\ldots) \in O_{(j,j')}^{\qq^*}\text{,}$$
where $q$ occurs at the $j$th coordinate and $q'$ occurs at the $j'$th coordinate, 
so that, if $h((\ldots,q,\ldots,q',\ldots))=\mathbf{c} \in P^*$, then  
$$\mathbf{c} \in O_{(j,j')}^{\pp^*}\text{,}$$
that is, $c_j \leq^{\pp} c_{j'}$.  We claim that $c_j=p$ and $c_{j'}=p'$, 
which implies $e(q)=p <^{\pp} p'=e(q')$.  
Indeed, since $h$ is a homomorphism, it is the case that $\mu_j(q)=\lambda_j(c_j)=k$ 
and $\mu_{j'}(q')=\lambda_{j'}(c_{j'})=k'$, because 
there is both a $R_{(j,k)}$ loop and a $R_{(j',k')}$ loop over $(\ldots,q,\ldots,q',\ldots)$ in $\qq^*$.  
Then, by Claim~\ref{cl:cl1} and the definition of $e$, it is the case that $c_j=e(q)=p$ 
and $c_{j'}=e(q')=p'$.

$q \parallel^\qq q'$ implies $e(q) \parallel^\pp e(q')$:  Let $\mu_j(q)=k$ and $\mu_{j'}(q')=k'$.  
We have that 
$$(\ldots,q,\ldots,q',\ldots) \in I_{\{j,j'\}}^{\qq^*}\text{,}$$
where $q$ occurs at the $j$th coordinate and $q'$ occurs at the $j'$th coordinate, 
so that, if $h((\ldots,q,\ldots,q',\ldots))=\mathbf{c} \in P^*$, then  
$$\mathbf{c} \in I_{\{j,j'\}}^{\pp^*}\text{,}$$
that is, $c_j\parallel^{\pp} c_{j'}$.  We claim that $c_j=e(q)$ and $c_{j'}=e(q')$, 
which implies $e(q) \parallel^{\pp} e(q')$.  
Indeed, since $h$ is a homomorphism, it is the case that $\lambda_j(c_j)=k$ 
and $\lambda_{j'}(c_{j'})=k'$, because 
there is both a $R_{(j,k)}$ loop and a $R_{(j',k')}$ loop over $(\ldots,q,\ldots,q',\ldots)$ in $\qq^*$.  
Then, by Claim~\ref{cl:cl1} and the definition of $e$, it is the case that $c_j=e(q)$ 
and $c_{j'}=e(q')$.
\end{proof}}

\newcommand{\excompilhom}[0]{
\begin{example}\label{ex:compilhom}
Let $\qq$ and $\pp$ be the posets in Example~\ref{ex:chainpartition}, 
so that $\qq$ embeds into $\pp$ via the map $e \colon Q \to P$ defined 
in the example (see Figure~\ref{fig:chainpartition}).  Let 
$\qq^*=\textup{compil}(\qq,\dd_{1},\dd_{2},\mu_1,\mu_{2})$ 
and $\pp^*=\textup{compil}(\pp,\cc_{1},\cc_{2},\lambda_1,\lambda_{2})$ 
be the structures in Example~\ref{ex:compilp}, respectively compiling $\qq$ and $\pp$.  
The homomorphism $h \colon Q^* \to P^*$, corresponding to the embedding $e \colon Q \to P$ as by (the forward direction of) 
Lemma~\ref{lemma:correct}, 
is depicted in Figure~\ref{fig:compilhom}.

\begin{figure}[h]
\centering
\includegraphics[scale=.2]{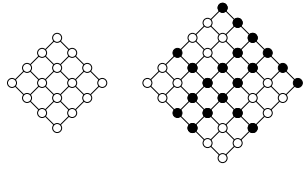}
\caption{The structures $\qq^*$ (left) and $\pp^*$ (right) in Example~\ref{ex:compilhom}.  
The white points in $P^*$ form the image of the homomorphism $h \colon Q^* \to P^*$ in Example~\ref{ex:compilhom}.  
It is possible to check that $h$ is a homomorphism by direct inspection of Figure~\ref{fig:compilq} and Figure~\ref{fig:compilp}.}
\label{fig:compilhom}
\end{figure}  
\end{example}}



\begin{proof}[Proof of Lemma~\ref{lemma:correct}]
$\textit{(i)} \Rightarrow \textit{(ii)}$:  Let $e \colon Q\to P$ be an embedding of $\qq$ into $\pp$.  

\longshort{\begin{claim}}{\begin{claim}}
\label{claim:forwclaim1}
There exist $w' \in \mathbb{N}$ such that $\textup{width}(\qq) \leq w' \leq w$, 
a subtuple $(i_1,i_2,\ldots,i_{w'})$ of $(1,2,\ldots,w)$, 
and a chain partition $(\dd_{i_1},\ldots,\dd_{i_{w'}})$ of $\qq$ 
such that, for all $j \in [w']$, $e(D_{i_j})=\{ e(d) \mid d \in D_{i_j} \} \subseteq C_{i_j}$.  
\end{claim}

\longshort{\pfforwclaima}{\pfforwclaima}

We let $\qq^*=\textup{compil}(\qq,\dd_{i_1},\ldots,\dd_{i_{w'}},\mu_1,\ldots,\mu_{w'})$ 
and $\pp^*=\textup{compil}(\pp,\cc_{i_1},\ldots,\cc_{i_{w'}},\lambda_1,\ldots,\lambda_{w'})$ 
be the compilations of $\qq$ and $\pp$ respectively, given by the colorings 
$(\mu_1,\ldots,\mu_{w'})$ and $(\lambda_1,\ldots,\lambda_{w'})$ defined as follows.  

For each $j \in [w']$, let $\Lambda_j$ be a $|D_{i_j}|$-perfect family of hash functions from $C_{i_j}$ to $[|D_{i_j}|]$.  
Let $j \in [w']$.  Let $\lambda_j \in \Lambda_j$ be such that $\lambda_j|_{e(D_{i_j})}$ is injective; 
indeed such a $\lambda_j$ exists, because 
$e(D_{i_j})$ is a subset of $C_{i_j}$ of cardinality $|D{i_j}|$ (as $e$ is injective), 
and $\Lambda_j$ is a $|D_{i_j}|$-perfect family of hash functions from $C_{i_j}$ to $[|D_{i_j}|]$.  
Let $e(D_{i_j})=\{c_1,\ldots,c_{|D_{i_j}|}\}$.  Let $\lambda_j(c_1)=k_1,\ldots,\lambda_j(c_{|D_{i_j}|})=k_{|D_{i_j}|}$. 
We let $\mu_j$ be such that $\mu_j(e^{-1}(c_i))=k_i$ 
for all $i \in [|D_{i_j}|]$.  Clearly, $\mu_j$ is a bijection from $D_{i_j}$ to $[|D_{i_j}|]$.  

The following claim settles the forward direction.

\longshort{\begin{claim}}{\begin{claim}}
\label{claim:forwclaim}
The function $h \colon Q^* \to P^*$ defined 
by $$h((d_1,\ldots,d_{w'}))=(e(d_1),\ldots,e(d_{w'}))$$ 
for all $(d_1,\ldots,d_{w'}) \in Q^*$ maps $\mathbf{Q}^*$ homomorphically to $\mathbf{P}^*$.  
\end{claim}

\longshort{\pfforwclaim}{\pfforwclaim}

$\textit{(ii)} \Rightarrow \textit{(i)}$: Let $\qq^*$ and $\pp^*$ be specified as in the statement of the lemma, 
and let $h \colon Q^* \to P^*$ be a homomorphism from $\qq^*$ to $\pp^*$.  
We define a function $e \colon Q\to P$ as follows.  Below, 
$\mathbf{c}=(c_1,\ldots,c_{w'})$, 
$\mathbf{c}'=(c'_1,\ldots,c'_{w'})$, 
$\mathbf{c}''=(c''_1,\ldots,c''_{w'})$,  
$\mathbf{d}=(d_1,\ldots,d_{w'})$, and $\mathbf{d}'=(d'_1,\ldots,d'_{w'})$ are elements of $\qq^*$.

Let $q \in Q$.  Let $j \in [w']$ be such that $q \in D_{i_j}$.  

\longshort{\begin{claim}}{\begin{claim}}
\label{cl:cl1}
There exists a unique $p \in C_{i_j} \subseteq P$ such that:
\begin{itemize}
\item if $h(\mathbf{d})=\mathbf{c}$ and $d_j=q$, 
then $c_j=p$;
\item $\mu_j(q)=\lambda_j(p)$.
\end{itemize}
\end{claim}

\longshort{\pfcla}{\pfcla}

We define $e(q)=p\text{,}$
where $p \in P$ is the unique element identified by Claim~\ref{cl:cl1} relative to $q$. The following claim then settles the backwards direction.

\longshort{\begin{claim}}{\begin{claim}}
\label{cl:cl2}
$e$ embeds $\qq$ into $\pp$.
\end{claim}
\longshort{\pfclb}{\pfclb}

The statement is proved.\end{proof}

%
%
%
%
%
%
%
%

\longshort{\excompilhom}{\excompilhom}


\subsubsection{Algorithm}\label{sect:proof}

We are now ready to 
list the pseudocode of our main algorithm.  
The input is a pair $(\qq,\pp)$ of
posets.

\begin{tabbing}
\textsc{Algorithm}$(\qq,\pp)$\\
1 \quad \= \textbf{if} ($|P|<|Q|$ \textbf{or} $\textup{width}(\pp)<\textup{width}(\qq)$) \textbf{then reject}\\
2 \> \= $w \leftarrow \textup{width}(\pp)$ \\
3 \> \= compute a chain partition $(\cc_1,\ldots,\cc_w)$ of $\pp$ \\
4       \> \textbf{foreach} $1 \leq w' \leq w$,\\
       \> \> \quad \= subtuple $(i_1,\ldots,i_{w'})$ of $(1,\ldots,w)$,\\
       \> \> \> chain partition $(\dd_{i_1},\ldots,\dd_{i_{w'}})$ of $\qq$,\\
       \> \> \> coloring $(\mu_1,\ldots,\mu_{w'}) \in M_1 \times \cdots M_{w'}$ \textbf{do}\\
5       \> \quad \= \textbf{if} exists $j \in [w']$ such that $|C_{i_j}|<|D_{i_j}|$ \textbf{then reject}\\
6       \> \> $\qq^* \leftarrow \textup{compil}(\qq,\dd_{i_1},\ldots,\dd_{i_{w'}},\mu_1,\ldots,\mu_{w'})$\\
7       \>       \> \textbf{foreach} $j \in [w']$  \textbf{do}\\
8       \>       \> \quad \= $\Lambda_j \leftarrow$ $|D_{i_j}|$-perfect family of hashing functions\\
         \>       \> \>       $\ \ \ \ \ \ \ $ from $C_{i_j}$ to $[|D_{i_j}|]$\\
9       \>       \> \textbf{foreach} $(\lambda_1,\ldots,\lambda_{w'}) \in \Lambda_1 \times \cdots \times \Lambda_{w'}$  \textbf{do}\\
10       \>       \> \quad \= $\pp^* \leftarrow \textup{compil}(\pp,\cc_{i_1},\ldots,\cc_{i_{w'}},\lambda_1,\ldots,\lambda_{w'})$\\
11       \>       \>             \> \textbf{if} $\qq^* \in \textsc{Hom}(\pp^*)$ \textbf{then accept}\\
12      \> \textbf{reject}
\end{tabbing}

In Line~4, $M_j$ is the set of all bijections from $D_{i_j}$ to $[|D_{i_j}|]$, 
for all $j \in [w']$.  We conclude proving that the algorithm above 
has the desired properties, from which the main result of the section follows.

\longshort{\begin{theorem}}{\begin{theorem}[$\star$]}
\label{th:embfpt}
Let $\mathcal{P}$ be a class of 
posets of bounded width.  There exists an algorithm 
deciding any instance $(\qq,\pp)$ of $\textsc{Emb}(\mathcal{P})$ 
in $2^{O(k\log k)}\cdot n^{O(1)}$ time, 
where $n=|P|$ and $k=|Q|$.
\end{theorem}

\newcommand{\pfembfpt}[0]{
\begin{proof}
By Lemma \ref{lemma:correct}, \textsc{Algorithm} accepts if and only if there exists an embedding from $Q$ to $P$.  
Let us analyze its running time.  Let $n=|P|$ and $k=|Q|$.  In the rest of the analysis, 
we assume $k \leq n$; otherwise, the algorithm rejects in time $O(k+n)$ by the first test in Line 1.  

(The second test in) Line 1, and Lines 2-3 are feasible in time $n^{O(1)}$ by Theorem~\ref{th:felsner}. 
The loop between Line 4 and 10 executes at most $2^{O(k\log k)}$ times, 
and Lines 5-6 are feasible in time $n^{O(1)}$. 
The two loops in Lines 7-8 and 9-11 are feasible in time $2^{O(k)}\cdot n^{O(1)}$ by Theorem \ref{th:hash}; 
in particular, Line 11 executes in time $n^{O(1)}$ by 
Lemma \ref{lemma:compres} and Theorem \ref{th:semilpoly}. 
Hence the total running time is bounded above by $2^{O(k\log k)}\cdot n^{O(1)}$.
\end{proof}}

\longversion{\pfembfpt}


\begin{theorem}
\label{thm:EFOFPT}
Let $\mathcal{P}$ be a class of
posets of bounded width.  Then, 
$\textsc{MC}(\mathcal{P},\mathcal{FO}(\exists,\wedge,\vee,\neg))$ is fixed-parameter tractable 
(with single exponential parameter dependence).
\begin{proof}
Directly from Proposition~\ref{proposition:metodological}  
and Theorem~\ref{th:embfpt}.
\end{proof}
\end{theorem}

\subsection{Embedding is $\textup{W}[1]$-hard on Bounded Cover-Degree Posets}\label{sect:bdcover}

We construct a class $\mathcal{P}_{\textup{cover\textup{-}degree}}$ of bounded cover-degree 
posets 
such that $\textsc{Emb}(\mathcal{P}_{\textup{cover\textup{-}degree}})$ is $\textup{W}[1]$-hard.  
By Proposition~\ref{proposition:metodological}, it follows 
that $\textsc{MC}(\mathcal{P}_{\textup{cover\textup{-}degree}},\mathcal{FO}(\exists,\wedge,\vee,\neg))$ 
is $\textup{W}[1]$-hard.

Let $\mathbf{G}=(V,E^\mathbf{G})$ be a graph and let $V=[n]$.  Then $r(\mathbf{G})=\pp$ is the
poset defined as follows. The universe of $\pp$ is $P=\bigcup_{i \in [n]}P_i$ 
where, for all $i \in [n]$, 
\longshort{\begin{align*}
P_i = &\{ \bot_i,a_i,b_i,c_i,d_i,\top_i \} \cup \{ l_{(i,j)}, u_{(i,j)} \mid j \in [n], (i,j) \in E^{\mathbf{G}} \}\text{.}
\end{align*}}{\begin{align*}
P_i = &\{ \bot_i,a_i,b_i,c_i,d_i,\top_i \} \\
    & \cup \{ l_{(i,j)}, u_{(i,j)} \mid j \in [n], (i,j) \in E^{\mathbf{G}} \}\text{.}
\end{align*}}
The order is defined by the following. For each $i,j \in [n]$.:
\begin{itemize}
\item $a_i \prec^\pp b_i$, $a_i \prec^\pp c_i$, $b_i \prec^\pp d_i$, $c_i \prec^\pp d_i$, and $b_i \parallel^\pp c_i$; 
\item $\bot_i \prec^\pp l_{(i,1)} \prec^\pp  \cdots \prec^\pp l_{(i,n)} \prec^\pp a_i$;
\item $d_i \prec^\pp u_{(i,1)} \prec^\pp  \cdots \prec^\pp u_{(i,n)} \prec^\pp \top_i$;
\item $l_{(i,j)} \prec^\pp u_{(j,i)}$ if and only if $(i,j) \in E^{\mathbf{G}}$.
\end{itemize}
The construction satisfies the following properties.  Let $\mathbf{G} \in \mathcal{G}$:
%
\begin{enumerate}[label=\textit{(\roman*)}] 
\item since $\textup{cover\textup{-}degree}(r(\mathbf{G})) \leq 3$, 
the class $\mathcal{P}_{\textup{cover\textup{-}degree}}=\{ r(\mathbf{G}) \mid \mathbf{G} \in \mathcal{G} \}$ has bounded cover-degree;
\item $r(\mathbf{G})$ can be constructed in polynomial time;
\item for any $j,j'\in [n]$, $j \neq j'$, we have $\bot_j <^{\pp} \top_{j'}$ if and only if $(j,j')\in E^{\mathbf{G}}$.
\end{enumerate}

\longshort{\begin{proposition}}{\begin{proposition}[$\star$]}
\label{pr:harddegree}
$\textsc{Emb}(\mathcal{P}_{\textup{cover\textup{-}degree}})$ is $\textup{W}[1]$-hard.
\end{proposition}

\newcommand{\pfharddegree}[0]{
\begin{proof}

We give an fpt many-one reduction from the $\textsc{Clique}$ problem 
to $\textsc{Emb}(\mathcal{P}_{\textup{cover\textup{-}degree}})$, 
which suffices since $\textsc{Clique}$ is $\textup{W}[1]$-hard.  
The reader is advised to inspect Example~\ref{ex:degree}.

Let $(\mathbf{G},k)$ be an instance of $\textsc{Clique}$; 
the question is whether $\mathbf{G}$ contains a clique on $k \in \mathbb{N}$ vertices
Let $\mathbf{P}=r(\mathbf{G})$.  
We reduce to the instance $(\mathbf{Q}_k,\mathbf{P})$ of $\textsc{Emb}(\mathcal{P}_{\textup{cover\textup{-}degree}})$, 
where $\mathbf{Q}_k$ is the 
poset with 
universe $Q_k=\{ \bot_i,a_i,b_i,c_i,d_i,\top_i \mid i \in [k] \}$, 
uniquely determined by the following relations:
\begin{itemize}
\item $a_i \prec^{\qq_k} b_i$, $a_i \prec^{\qq_k} c_i$, $b_i \prec^{\qq_k} d_i$, $c_i \prec^{\qq_k} d_i$, and $b_i \parallel^{\qq_k} c_i$; 
\item $\bot_i \prec^{\qq_k} a_i$ and $d_i \prec^{\qq_k} \top_i$ for all $i \in [k]$;
\item $\bot_i \prec^{\qq_k} \top_{i'}$ for all $i,i' \in [k]$, $i \neq i'$.
\end{itemize}
We argue correctness (the complexity of the reduction is clear).  

If $\{j_1,\ldots,j_k\} \subseteq G$ induces a clique of size $k$ in $\mathbf{G}$, 
then $\qq_k$ embeds into $\pp$ by $q_i \mapsto q_{j_i}$ 
for all $q \in \{ \bot,a,b,c,d,\top \}$ and $i \in [k]$.  

Conversely, assume that $\qq_k$ embeds into $\pp$ via a mapping $e$. 
Let $i \in [k]$. We claim that there exists $j \in [n]$ such that 
$\{e(b_i),e(c_i)\}=\{b_j,c_j\}$.  Indeed, by construction, 
$b_i \parallel^{\qq_k} c_i$ and 
$a_i \leq^{\qq_k} b_i,c_i \leq^{\qq_k} d_i$.  
Note that any two incomparable elements $p' \in P_{j'}$ and $p'' \in P_{j''}$ with $j',j'' \in [n]$, 
$j' \neq j''$, lack a common upper bound or a common lower bound.  Hence, since $e$ is an embedding, 
$\{e(b_i),e(c_i)\} \subseteq P_j$ for some $j \in [n]$, which forces $\{e(b_i),e(c_i)\}=\{b_j,c_j\}$ 
because $b_j$ and $c_j$ are the only two incomparable elements in $P_j$.

We claim that $C=\{ j \mid \{b_j,c_j\} \cap \text{$e(Q_k) \neq \emptyset$} \}\subseteq V$ 
induces a clique of size $k$ in $\mathbf{G}$.  By the above, $|C|=k$.  Hence it 
suffices to show that $(j,j') \in E^{\mathbf{G}}$ for any $j,j' \in C$, $j \neq j'$.  
Let $i,i' \in [k]$, $i \neq i'$, be such that $\{e(b_i),e(c_i)\}=\{b_j,c_j\}$ 
and $\{e(b_{i'}),e(c_{i'})\}=\{b_{j'},c_{j'}\}$.
Since $e(\bot_i) <^{\pp} e(b_{i})$ and $e(\top_{i'}) >^{\pp} e(b_{i'})$, 
we obtain that $e(\bot_i)\in P_j$ and $e(\top_{i'})\in P_{j'}$ by construction. The embedding ensures $e(\bot_i) <^{\pp} e(\top_{i'})$ and so $(j,j') \in E^{\mathbf{G}}$ 
by the properties listed before the statement, concluding the proof.
\end{proof}
}

\longversion{\pfharddegree}

\newcommand{\exdegree}[0]{

\begin{example}\label{ex:degree}
Let $(\mathbf{G},k)$ be an instance of $\textsc{Clique}$, 
where $\mathbf{G}$ is the graph whose universe is $G=[4]$ 
and whose edge relation $E^\mathbf{G}$ is the symmetric closure of 
$\{(1,2),(1,3),(2,3),(2,4),(3,4)\}$, and $k=3$.  Then posets 
$\qq_k$ and $\pp$ in the proof of Proposition~\ref{pr:harddegree} 
are depicted in Figure~\ref{fig:qdegree} and Figure~\ref{fig:pdegree} respectively.
\end{example}}

\longversion{\exdegree}

\newcommand{\figqdegree}[0]{
\begin{figure}[h]
\centering
\includegraphics[scale=.2]{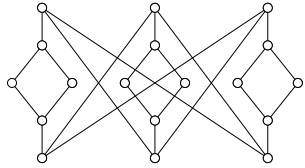}
\caption{The poset $\qq_k$ in the proof of Proposition~\ref{pr:harddegree}, 
where $k=3$ as in Example~\ref{ex:degree}.}
\label{fig:qdegree}
\end{figure}}

\longversion{\figqdegree}

\newcommand{\figpdegree}[0]{
\begin{figure}[h]
\centering
\includegraphics[scale=.2]{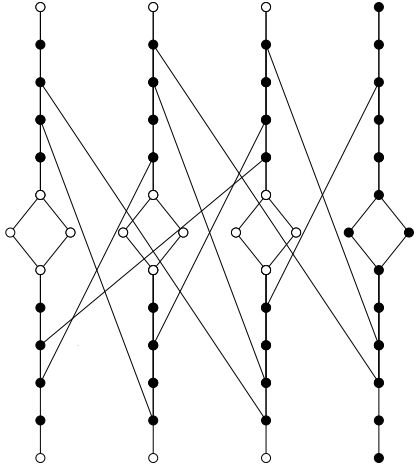}
\caption{The poset $\pp$ in the proof of Proposition~\ref{pr:harddegree}, 
where $\mathbf{G}$ is as in Example~\ref{ex:degree}.}
\label{fig:pdegree}
\end{figure}}

\longversion{\figpdegree}

\section{Classical Complexity}\label{sect:classical}
In this section, we study the classical complexity of 
the embedding problem on the targeted classes of posets, 
and we prove a tractability result of independent interest 
on bounded width posets. We first observe the following fact.  

\longshort{\begin{proposition}}{\begin{proposition}[$\star$]}
\label{th:wdtract}
Let $\mathcal{P}$ be a class of 
posets of bounded size.  Then, 
$\textsc{Emb}(\mathcal{P})$ is polynomial-time tractable.
\end{proposition}

\newcommand{\pfwdtract}[0]{

\begin{proof}
Let $s \in \mathbb{N}$ be such that $|\pp| \leq s$ for all $\pp \in \mathcal{P}$.  
Let $(\mathbf{Q},\mathbf{P})$ be an instance of $\textsc{Emb}(\mathcal{P})$.  
If $|Q|>|P|$, reject.  Otherwise, check whether one of the at most $s^s$ many 
mappings from $\mathbf{Q}$ to $\mathbf{P}$ is an embedding.  
\end{proof}}

\longshort{\pfwdtract}{}

Note that the above together with Proposition \ref{pr:exprcomplex} rules out a polynomial-time tractability analogue of Proposition \ref{proposition:metodological}.
The section is organized into three subsections, as follows.
\begin{itemize}
\item In Subsections \ref{sect:widthnphard} and \ref{sect:degreenphard}, we prove that the embedding problem is NP-hard on bounded width and bounded degree posets respectively. This implies that Proposition~\ref{th:wdtract} is tight with respect to the studied invariants.
\item In Subsection \ref{sect:wdtract}, we show how the ideas developed in Section \ref{sect:mainresults} may be used to obtain a polynomial-time tractable algorithm for the isomorphism of bounded width posets, an open problem in order theory \cite[p.~284]{CaspardLeclercMonjardet12}. 
\end{itemize}

\subsection{Embedding is $\textup{NP}$-hard on Bounded Width Posets}\label{sect:widthnphard}

In this subsection, we construct a class $\mathcal{P}$ of 
posets of bounded width 
such that $\textsc{Emb}(\mathcal{P})$ is $\textup{NP}$-hard, which immediately implies 
$\textup{NP}$-hardness of $\textsc{MC}(\mathcal{P},\mathcal{FO}(\exists,\wedge,\neg))$. 

The reduction, from the Boolean satisfiability problem (SAT), is technically involved.  
Intuitively, given a SAT instance $\phi$, 
we construct two bounded width posets $\qq_\phi$ and $\pp_\phi$.  
The two posets are such that, if $\phi$ is satisfiable, 
then $\qq_\phi$ embeds into $\pp_\phi$ \lq\lq nicely\rq\rq, 
in the sense that certain chains of $\qq_\phi$ embed 
into certain families of chains in $\pp_\phi$; 
conversely, 
every embedding of $\qq_\phi$ into $\pp_\phi$ must be nice in the above sense, 
and any nice embedding of $\qq_\phi$ into $\pp_\phi$ yields a satisfying assignment to $\phi$.
\newcommand{\exexa}[0]{
\begin{example}\label{ex:ex1}
Let $\phi(x_1,x_2,x_3,x_4)=\delta_1 \wedge \delta_2 \wedge \delta_3 \wedge \delta_4 \wedge \delta_5$, 
where 
$\delta_1=x_4 \vee \neg x_2$, 
$\delta_2=x_4 \vee \neg x_1$, 
$\delta_3=x_1 \vee \neg x_2$, 
$\delta_4=x_3 \vee \neg x_1$, and 
$\delta_5=\neg x_3 \vee x_2$.  Note that, for instance, 
$\phi$ is satisfied by $\{(x_1,0),(x_2,0),(x_3,0),(x_4,1)\}$.  

The poset $\qq_{\phi}$ is depicted in Figure~\ref{fig:qphi}, 
where the chain on the left is $Q^v_{\phi}$, 
the chain in the middle contains $Q^a_{\phi}$, 
and the chain on the right is $Q^c_{\phi}$.  
Thick edges represent chains of $|Q^a_{\phi}|$ elements.
\end{example}}
%
\longversion{\exexa}

\newcommand{\figqphi}[0]{
\begin{figure}[h]
\centering
\includegraphics[scale=.2]{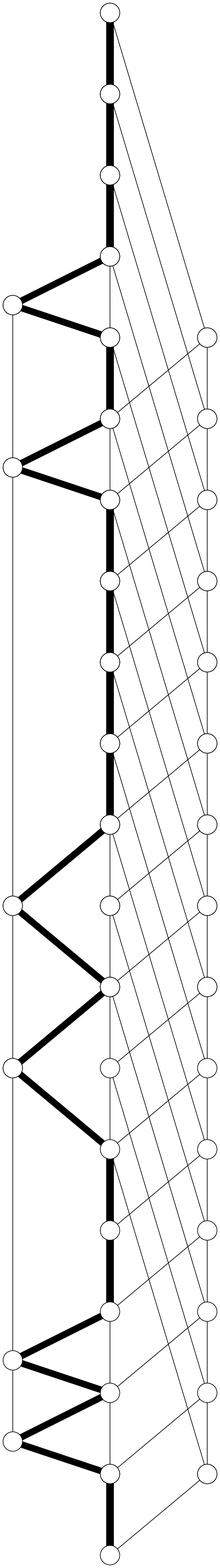}
\caption{The poset $\qq_{\phi}$ corresponding to $\phi \in \mathcal{S}$ in Example~\ref{ex:ex1}.}
\label{fig:qphi}
\end{figure}}

\longversion{\figqphi}

\newcommand{\exexb}[0]{
\begin{example}\label{ex:ex2}
Let $\phi \in \mathcal{S}$ be as in Example~\ref{ex:ex1}.  
Then, the poset $\pp_{\phi}$ is depicted in Figures~\ref{fig:pphi}-\ref{fig:pphi2}.  
The block on the left is $P^v_{\phi}$, 
the block in the middle contains $P^a_{\phi}$, 
and the block on the right is $P^c_{\phi}$.  
Thick edges represent chains of $|Q^a_{\phi}|$ elements.

The white points in $\pp_{\phi}$ form the image of the 
embedding $e \colon Q_{\phi} \to P_{\phi}$ of $\qq_{\phi}$ into $\pp_{\phi}$ corresponding 
to the $\phi$-satisfying assignment in Example~\ref{ex:ex1} as by (the easy direction of) Theorem~\ref{th:widthnphard}.  
It is possible to check that $e$ is an embedding by direct inspection of Figure~\ref{fig:qphi}, 
Figure~\ref{fig:pphi}, and Figure~\ref{fig:pphi2}.
\end{example}}

\newcommand{\figpphi}[0]{
\begin{figure}[h]
\centering
\includegraphics[scale=.2]{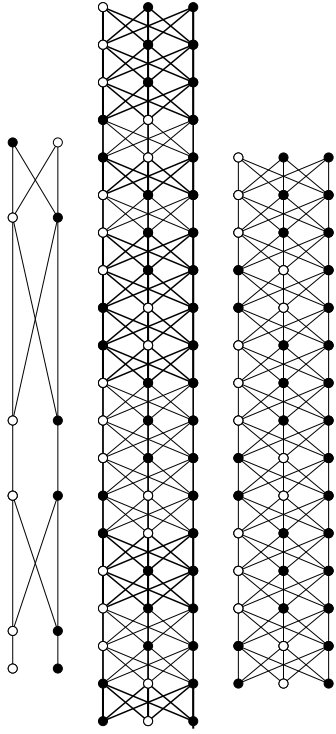}
\caption{Items (P1)-(P4) in the construction of 
poset $\pp_{\phi}$, where $\phi \in \mathcal{S}$ is as in Example~\ref{ex:ex1}.}
\label{fig:pphi}
\end{figure}}

\longversion{\figpphi}

\newcommand{\figpphib}[0]{
\begin{figure}[h]
\centering
\includegraphics[scale=.2]{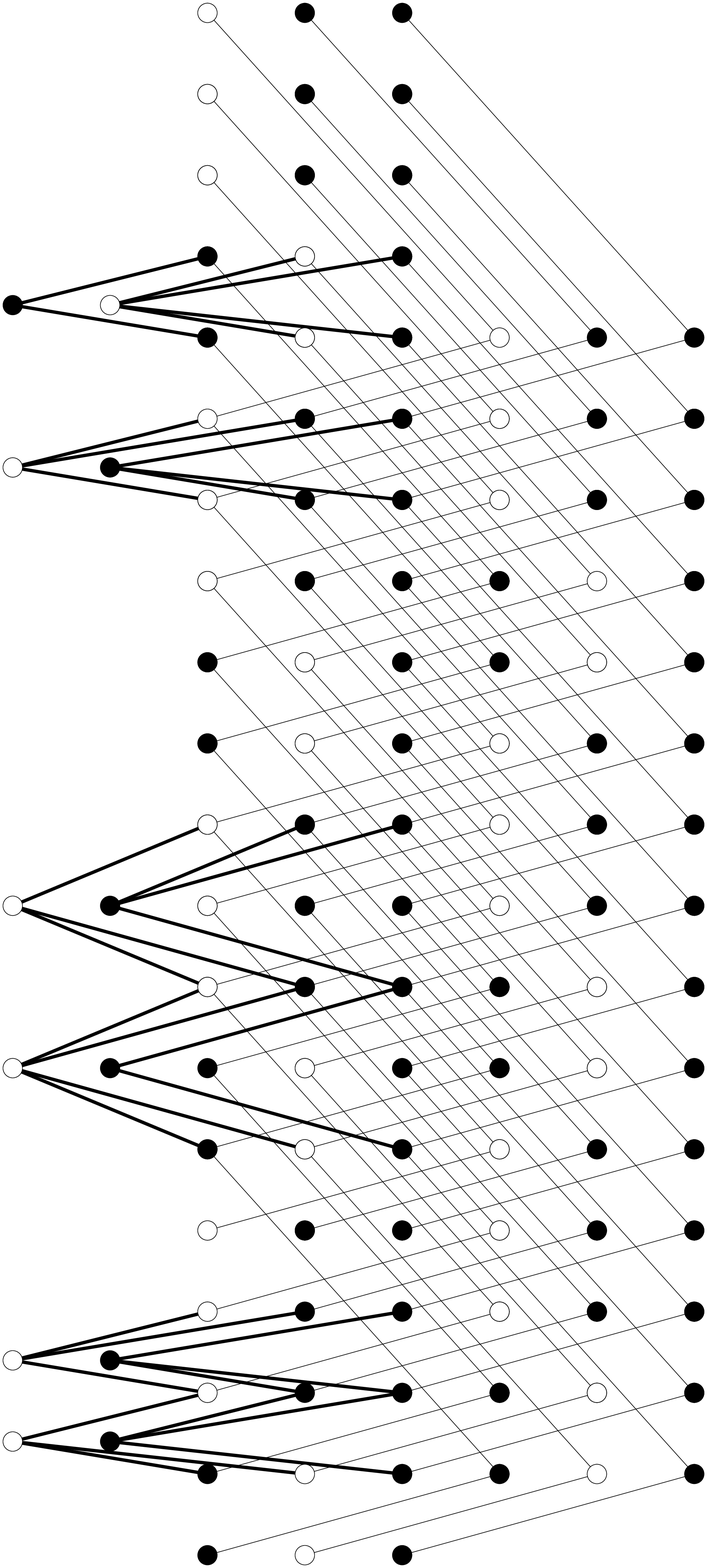}
\caption{Items (P5)-(P6) in the construction of poset $\pp_{\phi}$, 
where $\phi \in \mathcal{S}$ is as in Example~\ref{ex:ex1}.  The three points  
on the bottom, from left to right, represent $(\{(x_2,0),(x_4,0)\},(\delta_1,1))$, 
$(\{(x_2,1),(x_4,0)\},(\delta_1,1))$, and $(\{(x_2,1),(x_4,1)\},(\delta_1,1))$, 
that is, we display the satisfying assignments of $\delta_1$ 
in the order $\{(x_2,0),(x_4,0)\}$, $\{(x_2,1),(x_4,0)\}$, and $\{(x_2,1),(x_4,1)\}$.  
Similarly, we display the satisfying assignments: of $\delta_2$ in the order 
$\{(x_1,0),(x_4,0)\}$, $\{(x_1,1),(x_4,0)\}$, and $\{(x_1,1),(x_4,1)\}$; 
of $\delta_3$, in the order $\{(x_1,0),(x_2,0)\}$, $\{(x_1,1),(x_2,0)\}$, and $\{(x_1,1),(x_2,1)\}$; 
of $\delta_4$, in the order $\{(x_1,0),(x_2,0)\}$, $\{(x_1,1),(x_2,0)\}$, and $\{(x_1,1),(x_2,1)\}$; 
of $\delta_5$, in the order $\{(x_2,0),(x_3,0)\}$, $\{(x_2,0),(x_3,1)\}$, and $\{(x_2,1),(x_3,1)\}$.}
\label{fig:pphi2}
\end{figure}}

\longversion{\figpphib}

Let $\mathcal{S}$ be the class of propositional formulas in conjunctive form, containing at least $3$ clauses, 
where each clause contains at most $3$ literals; 
also, no clause contains a pair of complementary literals, 
and each variable occurs in at least two clauses.  
Let $\phi(x_1,\ldots,x_n)=\delta_1 \wedge \cdots \wedge \delta_m$ be in $\mathcal{S}$.  
For $i \in [n]$ and $j \in [m]$, we write $x_i \in \delta_j$ if a literal on variable $x_i$ occurs in clause $\delta_j$, 
and we let $\textup{var}(\delta_j)=\{ x_i \mid i \in [n], x_i \in \delta_j \}$.   

\longshort{We proceed in two stages (recall Example~\ref{ex:ex1}). }{We proceed in two stages.}
First, we define a poset $\qq_{\phi}$ as follows.  
The universe $Q_{\phi}$ contains 
$Q_{\phi}^a = \{ (\delta_i,j) \mid i \in [m], j \in [n]\}$, 
$Q_{\phi}^c = \{ (\delta'_i,j) \mid i \in [m], j \in [n-1]\}$, 
\begin{align*}
Q_{\phi}^v &=\left\{ (x_i,(j,j')) \suchthat 
\begin{array}{c}
i \in [n], x_i \in \delta_j, x_i \in \delta_{j'},j<j'\textup{,}\\
\textup{and } x_i \not\in \delta_{j''} \textup{ for all } j<j''<j'
\end{array} \right\}\text{,}
\end{align*}
and a set $Q_{\phi}^l$ of auxiliary elements introduced below. 

For $q,q' \in Q_{\phi}$, we let $\ll^{\qq_{\phi}}$ 
denote the fact that, in the order of $\qq_{\phi}$, 
there is a chain of $|Q^a_{\phi}|$ fresh auxiliary elements, contained in $Q^l_{\phi}$, between $q$ and $q'$.  
The order relation of $\qq_{\phi}$ is defined by the following cover relations:
\begin{itemize}
\item[(Q1)] for all $(\delta_i,j),(\delta_{i+1},{j}) \in Q_{\phi}^a$:  
if $i+1 \leq i'$ where $i'$ is the minimum in $[m]$ such that $x_j \in \delta_{i'}$, 
then $(\delta_i,j) \ll^{\qq_{\phi}} (\delta_{i+1},{j})$;  
if $i' \leq i$ where $i'$ is the maximum in $[m]$ such that $x_j \in \delta_{i'}$, 
then $(\delta_i,j) \ll^{\qq_{\phi}} (\delta_{i+1},{j})$;  
otherwise, $(\delta_i,j) \prec^{\qq_{\phi}} (\delta_{i+1},{j})$;
\item[(Q2)] $(\delta_m,j) \ll^{\qq_{\phi}} (\delta_{1},{j+1})$, for $(\delta_m,j),(\delta_{1},{j+1}) \in Q_{\phi}^a$; 
\item[(Q3)] $(\delta'_i,j) \prec^{\qq_{\phi}} (\delta'_{i+1},{j})$, for $(\delta'_i,j),(\delta'_{i+1},{j}) \in Q_{\phi}^c$; 
\item[(Q4)] $(\delta'_m,j) \prec^{\qq_{\phi}} (\delta'_{1},{j+1})$, for $(\delta'_m,j),(\delta'_{1},{j+1}) \in Q_{\phi}^c$; 
\item[(Q5)] $(x_i,(j,j')) <^{\qq_{\phi}} (x_i,(j',j''))$, 
for 
$(x_i,(j,j')),(x_i,(j',j'')) \in Q_{\phi}^v$; 
\item[(Q6)] $(x_i,(j,j')) <^{\qq_{\phi}} (x_{i+1},(k,k'))$, 
for 
$(x_i,(j,j')),(x_i,(k,k')) \in Q_{\phi}^v$ where 
$j'$ is maximum in $[m]$ such that $x_i \in \delta_{j'}$ 
and $k$ is minimum in $[m]$ such that $x_{i+1} \in \delta_{k}$.
\item[(Q7)] $(\delta_i,j) \prec^{\qq_{\phi}} (\delta'_i,j) \prec^{\qq_{\phi}} (\delta_i,j+1)$, 
for all $(\delta_i,j),(\delta_i,j+1) \in Q_{\phi}^a$ and $(\delta'_i,j) \in Q_{\phi}^c$;
\item[(Q8)] $(\delta_i,j) \ll^{\qq_{\phi}} (x_j,(i,i')) \ll^{\qq_{\phi}} (\delta_{i'},j)$, 
for all $(\delta_i,j),(\delta_{i'},j) \in Q_{\phi}^a$ and $(x_j,(i,i')) \in Q_{\phi}^v$.
\end{itemize}

Second, we define the poset $\pp_{\phi}=r(\phi)$, 
using $\qq_{\phi}$ as a basis, as follows.  
The universe $P_{\phi}$ is the union of  
\begin{align*}
P_{\phi}^a &= \bigcup_{(\delta_i,j) \in Q_{\phi}^a} \{ (f,(\delta_i,j)) \mid f \in \{0,1\}^{\textup{var}(\delta_i)} \textup{ satisfies } \delta_i \}\text{,}\\ 
P_{\phi}^c &= \bigcup_{(\delta'_i,j) \in Q_{\phi}^c} \{ (f,(\delta'_i,j)) \mid f \in \{0,1\}^{\textup{var}(\delta_i)} \textup{ satisfies } \delta_i \}\text{,}\\
P_{\phi}^v &=\bigcup_{(x_i,(j,j')) \in Q_{\phi}^v} \{ (x_i,(j,j')), (\neg x_i,(j,j'))\}\text{,}
\end{align*}
and a set $P_{\phi}^l$ of auxiliary elements introduced below. 

Again, for $p,p' \in P_{\phi}$, we let $\ll^{\pp_{\phi}}$ 
denote the fact that, in the order of $\pp_{\phi}$, 
there is a chain of $|Q^a_{\phi}|$ fresh auxiliary elements, contained in $P_{\phi}^l$, between $p$ and $p'$.   
The order relation of $\pp_{\phi}$ is defined by the following 
cover relation:
\begin{itemize}
\item[(P1)] for all $(f,(\delta_i,j)),(f',(\delta_{i'},{j'})) \in P_{\phi}^a$, 
$(f,(\delta_i,j)) \prec^{\pp_{\phi}} (f',(\delta_{i+1},{j}))$ 
if and only if 
$(\delta_i,j) \prec^{\qq_{\phi}} (\delta_{i'},{j'})$, 
and 
$(f,(\delta_i,j)) $\\
$\ll^{\pp_{\phi}} (f',(\delta_{i+1},{j}))$ 
if and only if 
$(\delta_i,j) \ll^{\qq_{\phi}} (\delta_{i'},{j'})$;

\item[(P2)] for all $(f,(\delta'_i,j)),(f',(\delta'_{i'},{j'})) \in P_{\phi}^c$, 
$(f,(\delta'_i,j)) \prec^{\pp_{\phi}} (f',(\delta'_{i+1},{j}))$ if and only if $(\delta'_i,j) \prec^{\qq_{\phi}} (\delta'_{i'},{j'})$;

\item[(P3)] for all $(x_i,(j,j'))$, $(\neg x_i,(j,j'))$, $(x_i,(j',j''))$, $(\neg x_i,(j',j''))$ in $P_{\phi}^v$, 
$(x_i,(j,j')) \prec^{\pp_{\phi}} (x_i,(j',j''))$ and 
$(\neg x_i,(j,j')) \prec^{\pp_{\phi}} (\neg x_i,(j',j''))$ 
if and only if $(x_i,(j,j')) \prec^{\qq_{\phi}} (x_i,(j',j''))$; 

\item[(P4)] for all $(x_i,(j,j'))$, $(\neg x_i,(j,j'))$, $(x_{i+1},(k,k'))$, and\\ $(\neg x_{i+1},(k,k'))$ in $P_{\phi}^v$, 
$(x_i,(j,j'))\longversion{\!} \prec^{\pp_{\phi}}\longversion{\!}(x_{i+1},(k,k'))$,  
$(x_i,(j,j'))\longversion{\!} \prec^{\pp_{\phi}}\longversion{\!}(\neg x_{i+1},(k,k'))$,  
$(\neg x_{i},(j,j')) \prec^{\pp_{\phi}}$\\
$(x_{i+1},(k,k'))$, 
and $(\neg x_{i},(j,j')) \prec^{\pp_{\phi}}(\neg x_{i+1},(k,k'))$, 
if and only if $(x_i,(j,j')) \prec^{\qq_{\phi}} (x_{i+1},(k,k'))$.

\item[(P5)] for all $(f,(\delta_i,j)), (f,(\delta_i,j+1)) \in P_{\phi}^a$ and $(f,(\delta'_i,j)) \in P_{\phi}^c$, 
$(f,(\delta_i,j)) \prec^{\pp_{\phi}} (f,(\delta'_i,j)) \prec^{\pp_{\phi}} (f,(\delta_i,j+1))$
if and only if 
$(\delta_i,j) \prec^{\qq_{\phi}} (\delta'_i,j) \prec^{\qq_{\phi}} (\delta_i,j+1)$;

\item[(P6)] for all 
$(f,(\delta_i,j)),(f',(\delta_{i'},j)),(g,(\delta_i,j)),(g',(\delta_{i'},j)) \in P_{\phi}^a$ 
and $(x_j,(i,i')),(\neg x_j,(i,i')) \in P_{\phi}^v$, it holds that 
$(f,(\delta_i,j)) \ll^{\pp_{\phi}} (x_j,(i,i')) \ll^{\pp_{\phi}} (f',(\delta_{i'},j))$ 
and $(g,(\delta_i,j)) \ll^{\pp_{\phi}} (\neg x_j,(i,i'))$\\$ \ll^{\pp_{\phi}} (g',(\delta_{i'},j))$ 
if and only if
$(\delta_i,j) \ll^{\qq_{\phi}} (x_j,(i,i')) \ll^{\qq_{\phi}} (\delta_{i'},j)$, 
$f(x_j)=f'(x_j)=1$, and $g(x_j)=g'(x_j)=0$.  
\end{itemize}
Note that $\textup{width}(\qq_{\phi}) \leq 4$ and $\textup{width}(\mathbf{P}_{\phi}) \leq 2^2+7^2+7^2=102$ 
for all $\phi \in \mathcal{S}$ 
(we remark that this width bound may be improved at the cost of a more complicated construction).  Hence 
$\mathcal{P}_{\textup{width}}=\{ r(\phi) \mid \phi \in \mathcal{S} \} $ 
has bounded width.

\longshort{\begin{theorem}}{\begin{theorem}[$\star$]}
\label{th:widthnphard} 
$\textsc{Emb}(\mathcal{P}_{\textup{width}})$ is $\textup{NP}$-hard.  
\end{theorem}

\newcommand{\pfwidthnphard}[0]{
\begin{proof}
We give a polynomial-time many-one reduction from the satisfiability problem over 
$\mathcal{S}$ to the problem $\textsc{Emb}(\mathcal{P}_{\textup{width}})$, 
which suffices since the source problem is $\textup{NP}$-hard.  

The reduction maps an instance $\phi \in \mathcal{S}$ of the satisfiability problem, 
say $\phi(x_1,\ldots,x_n)=\delta_1 \wedge \cdots \wedge \delta_m$,  
to the instance $(\mathbf{Q}_{\phi},\mathbf{P}_{\phi})$ of $\textsc{Emb}(\mathcal{P}_{\textup{width}})$, 
where $\mathbf{Q}_{\phi}$ and $\mathbf{P}_{\phi}$ are constructed as above.  
The reduction is clearly polynomial-time computable.  We prove that the reduction is correct.

If $\phi$ is satisfiable, then let $g \colon \{x_1,\ldots,x_n\} \to \{0,1\}$ 
be a satisfying assignment.  We define a function $e \colon Q_{\phi} \to P_{\phi}$ 
as follows.  Let $q \in Q_{\phi}$.  Then:
\begin{itemize}
\item If $q=(\delta_i,j) \in Q_{\phi}^a$, 
then $e(q)=(f,(\delta_i,j)) \in P_{\phi}^a$ 
if and only if $g|_{\textup{var}(\delta_i)}=f$.
\item If $q=(\delta'_i,j) \in Q_{\phi}^c$, 
then $e(q)=(f,(\delta'_i,j)) \in P_{\phi}^c$ 
if and only if $g|_{\textup{var}(\delta_i)}=f$. 
\item If $q=(x_i,(j,j')) \in Q_{\phi}^v$, 
then $e(q)=(x_i,(j,j')) \in P_{\phi}^v$ if $g(x_i)=1$, 
and $e(q)=(\neg x_i,(j,j')) \in P_{\phi}^v$ if $g(x_i)=0$. 
\item If $q \in Q_{\phi}^l$, 
then let $q',q'' \in Q_{\phi}$ 
and $q_1,\ldots,q_{|Q^a_{\phi}|} \in Q_{\phi}^l$ 
be such that $q' \prec^{\qq_{\phi}} q_1 \prec^{\qq_{\phi}} \cdots \prec^{\qq_{\phi}} q_{|Q^a_{\phi}|} \prec^{\qq_{\phi}} q''$ 
and $q=q_i$ for $i \in [|Q^a_{\phi}|]$.  By construction, there exist 
$p_1,\ldots,p_{|Q^a_{\phi}|} \in P_{\phi}^l$ 
such that $e(q') \prec^{\pp_{\phi}} p_1 \prec^{\pp_{\phi}} \cdots \prec^{\pp_{\phi}} p_{|Q^a_{\phi}|} \prec^{\qq_{\phi}} e(q'')$.  
Then, $e(q)=e(q_i)=p_i$.   
\end{itemize}
It is easy to check that $e$ embeds $\qq_{\phi}$ into $\pp_{\phi}$. 

Conversely, let $e \colon Q_{\phi} \to P_{\phi}$ 
be an embedding of $\qq_{\phi}$ into $\pp_{\phi}$.  

\begin{claim}\label{claim:empart}
$e(Q_{\phi}^a) \subseteq P_{\phi}^a$, 
$e(Q_{\phi}^v) \subseteq P_{\phi}^v$, 
$e(Q_{\phi}^c) \subseteq P_{\phi}^c$. 

\begin{proof}[Proof of Claim~\ref{claim:empart}]  
Let $Q^*=\{ q \in Q_{\phi}^a \mid \text{$q$ is comparable to all elements in $Q_{\phi}^v$} \}$.  
Note that, by construction, $\textup{depth}(\qq_{\phi})=|Q_{\phi}^v \cup Q_{\phi}^l \cup Q^*|=d$, 
and the chain $Q_{\phi}^v \cup Q_{\phi}^l \cup Q^*$ is the unique chain 
whose size equals $d$.  
In the poset $\qq_{\phi}$ depicted in Figure~\ref{fig:qphi}, 
$Q^*$ contains exactly the elements of the middle chain hit by a thick edge, 
and the chain $Q_{\phi}^v \cup Q_{\phi}^l \cup Q^*$ is represented by the thick edges.  
Moreover, by construction again, $\textup{depth}(\pp_{\phi})=d$, 
and the only chains in $\pp_{\phi}$ whose size equals $d$ 
force the embedding to satisfy $e(Q_{\phi}^v) \subseteq P_{\phi}^v$ 
and $e(Q^*) \subseteq P_{\phi}^a$.  

We now prove that $e(Q_{\phi}^a \setminus Q^*) \subseteq P_{\phi}^a$, 
which, together with the above, yields $e(Q_{\phi}^a) \subseteq P_{\phi}^a$.  
Indeed, let $q \in Q_{\phi}^a \setminus Q^*$.  
Let $q',q'' \in Q^*$ be such that $q' <^{\qq_{\phi}} q <^{\qq_{\phi}} q''$ 
and there do not exist $r',r'' \in Q^*$ such that $q' <^{\qq_{\phi}} r' <^{\qq_{\phi}} q$ 
or $q <^{\qq_{\phi}} r'' <^{\qq_{\phi}} q''$.  
In Figure~\ref{fig:qphi}, if, for instance, $q$ is the $8$th lowest element in the middle chain, 
then $q'$ and $q''$ are respectively the $6$th and $9$th lowest elements in the middle chain.  
Let $S=\{ p \in P_{\phi} \mid e(q') <^{\pp_{\phi}} p <^{\pp_{\phi}} e(q'') \}$, 
so that $e(q) \in S$, because $e$ is an embedding.  By the above, 
$S \cap (P_{\phi}^v \cup P_{\phi}^l) \subseteq e(Q_{\phi}^v \cup Q_{\phi}^l \cup Q^*)$, 
therefore $e(q) \in S \setminus (P_{\phi}^v \cup P_{\phi}^l)$.  Moreover, 
the distance between $e(q')$ and $e(q'')$ in $\pp_{\phi}$ is strictly 
less than $m$, therefore $S \cap P_{\phi}^c=\emptyset$.  It follows that 
$e(q) \in S \setminus (P_{\phi}^v \cup P_{\phi}^l \cup P_{\phi}^c)$, 
that is, $e(q) \in P_{\phi}^a$.

Finally, we prove that $e(Q_{\phi}^c) \subseteq P_{\phi}^c$.  
Indeed, let $q \in Q_{\phi}^c$.  By construction, 
there exist $m+1$ elements $q_0,\ldots,q_m \in Q_{\phi}^a$ 
such that $q_0 <^{\qq_{\phi}} \cdots <^{\qq_{\phi}} q_{m}$, 
$q_0 <^{\qq_{\phi}} q <^{\qq_{\phi}} q_m$, 
and $q$ is incomparable to $q_1,\ldots,q_{m-1}$  in $\qq_{\phi}$.   
By the above, $e(q_0),\ldots,e(q_{m}) \in P_{\phi}^a$.  
As $e$ is an embedding, 
$e(q_0) <^{\pp_{\phi}} \cdots <^{\pp_{\phi}} e(q_{m})$, 
$e(q_0) <^{\pp_{\phi}} e(q) <^{\pp_{\phi}} e(q_m)$, 
and $e(q)$ is incomparable to $e(q_1),\ldots,e(q_{m-1})$ in $\pp_{\phi}$.  
By inspection of the construction, 
we now prove that $e(q) \not\in P_{\phi}^a \cup P_{\phi}^v \cup P_{\phi}^l$, 
which implies $e(q) \in P_{\phi}^c$ as desired.  
If $e(q) \in P_{\phi}^a$, 
then $e(q)$ is incomparable to at most $1$ element 
among $e(q_1),\ldots,e(q_{m-1})$, 
which implies $e(q) \not\in P_{\phi}^a$ since $m>2$.  
If $e(q) \in P_{\phi}^v \cup P_{\phi}^l$, 
then $e(q)$ is incomparable to at most $m-2$ elements 
among $e(q_1),\ldots,e(q_{m-1})$, 
which implies $e(q) \not\in P_{\phi}^v \cup P_{\phi}^l$.
\end{proof}
\end{claim}

The previous three properties uniquely determine the behavior of $e$ over $Q_{\phi}^l$.  
Next, we state two facts which follow from the embedding and specific properties of the construction of $\mathbf{Q}$ and $\mathbf{P}$.
\begin{itemize}
\item Items (Q1)-(Q4) and (Q7) on one hand and (P1)-(P2) and (P5) on the other hand enforce the following:
for all $i \in [m]$ and $j \in [n]$, 
there exists a unique $f \in \{0,1\}^{\textup{var}(\delta_i)}$ 
such that for all $(\delta_i,j),(\delta'_i,j) \in Q_{\phi}^a \cup Q_{\phi}^c$ 
it holds that $e((\delta_i,j))=(f,(\delta_i,j))$ and $e((\delta'_i,j))=(f,(\delta'_i,j))$. 
\item Items (Q5)-(Q6) and (Q8) on one hand  
and (P3)-(P5) and (P6) on the other hand enforce the following:
for all $i,i' \in [m]$, $i \neq i'$, and $j \in [n]$ 
such that $x_j \in \textup{var}(\delta_i) \cap \textup{var}(\delta_{i'})$, 
it holds that if $e((\delta_i,j))=(f,(\delta_i,j))$ and $e((\delta_{i'},j))=(f',(\delta_{i'},j))$, 
then $f(x_j)=f'(x_j)$.  
\end{itemize}

Therefore the union of all the assignments $f$ 
such that $e((\delta_i,\cdot))=(f,(\delta_i,\cdot))$, taken over all $i \in [m]$, 
defines an assignment $g \colon \{x_1,\ldots,x_n\} \to \{0,1\}$, 
and moreover $g$ satisfies $\phi$. This concludes the proof
\end{proof}}

\longversion{\pfwidthnphard}

\subsection{Embedding is $\textup{NP}$-hard on Bounded Degree Posets}\label{sect:degreenphard}

\longshort{We reduce from the satisfiability problem. Let $\mathcal{S}$ be the class of propositional formulas in conjunctive form, 
where each clause contains exactly $3$ pairwise non-complementary literals (for notational convenience, 
but the construction works even if relaxed to at most $3$ literals, 
which we use for illustration purposes in the examples).}{We reduce from the satisfiability problem. Let $\mathcal{S}$ be the class of propositional formulas in conjunctive form, 
where each clause contains exactly $3$ pairwise non-complementary literals.} 

\newcommand{\exdegreenphard}[0]{
\begin{example}\label{ex:degreenphard}
Let $\phi(x_1,x_2,x_3)=\delta_1 \wedge \delta_2 \wedge \delta_3$, 
where 
$\delta_1=x_1 \vee \neg x_2$, 
$\delta_2=x_3 \vee \neg x_1$, and 
$\delta_3=\neg x_3 \vee x_2$.  Note that, for instance, 
$\phi$ is satisfied by $\{(x_1,0),(x_2,0),(x_3,0)\}$.  

The poset $\qq_{\phi}$ is depicted in Figure~\ref{fig:qdegreenphard}, 
where $Q_0$, $Q_1$, and $Q_2$ form respectively the bottom, middle, and top layers of the diagram; 
poset $\pp_{\phi}$ is similarly displayed in Figure~\ref{fig:pdegreenphard}.  
The white points in $\pp_{\phi}$ form the image of the 
embedding $e \colon Q_{\phi} \to P_{\phi}$ of $\qq_{\phi}$ into $\pp_{\phi}$ corresponding 
to the satisfying assignment above as by (the easy direction of) Theorem~\ref{th:degreenphard}.  
\end{example}}

The idea of the reduction is the following.  We encode a formula in $\mathcal{S}$ by a poset $\pp$, 
whose universe partitions into three blocks, $P_0$, $P_1$ and $P_2$. 
The set $P_1$ contains several groups of $7$ elements, 
where each element corresponds to one possible satisfying assignment of a clause, and the embedding 
encodes an assignment for the whole formula by forcing us to choose one element out of each group. 
The set $P_2$ ensures that each assignment chosen by the embedding is consistent for each pair of clauses. 
To preserve bounded degree while ensuring the consistency of each pair of clauses, 
it is necessary to use many groups in $P_1$ for each clause. Finally, 
$P_0$ ensures that each choice made by the embedding for a given clause is consistent across all groups corresponding to that clause.  
\longversion{\exdegreenphard}

\newcommand{\figqdegreenphard}[0]{
\begin{figure}[h]
\centering
\includegraphics[scale=.2]{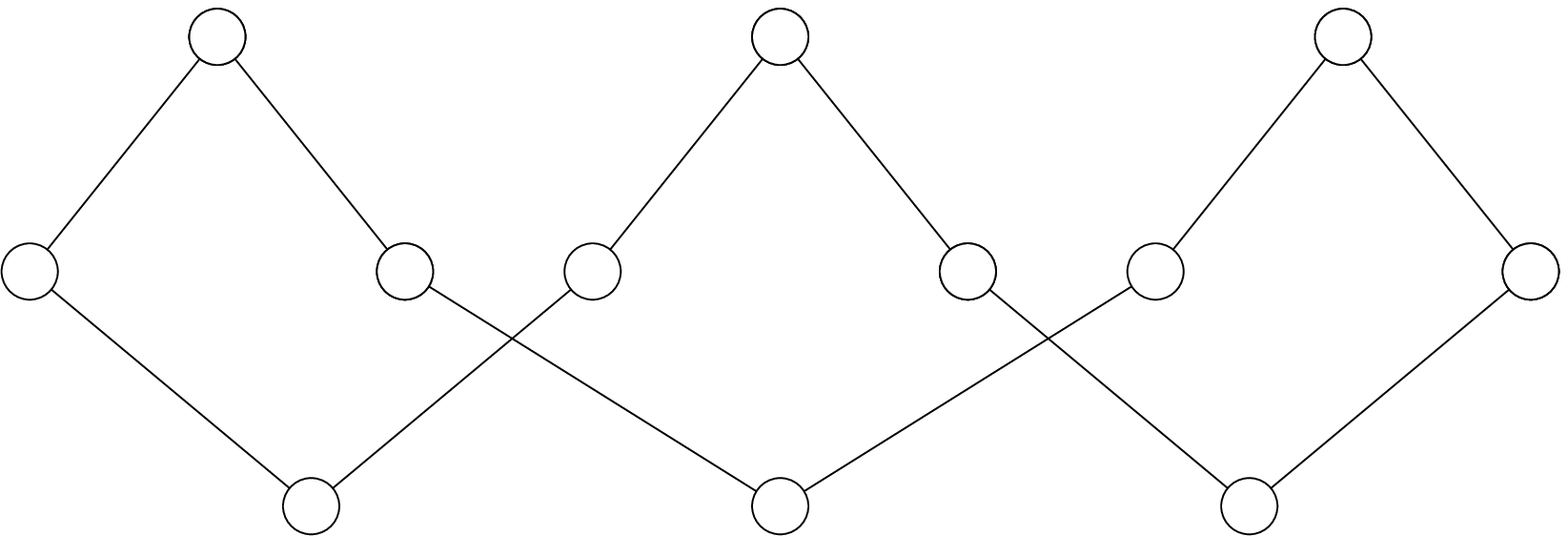}
\caption{The poset $\qq_{\phi}$ corresponding to $\phi \in \mathcal{S}$ in Example~\ref{ex:degreenphard}.}
\label{fig:qdegreenphard}
\end{figure}}

\longversion{\figqdegreenphard}

\newcommand{\figpdegreenphard}[0]{
\begin{figure}[h]
\centering
\includegraphics[scale=.2]{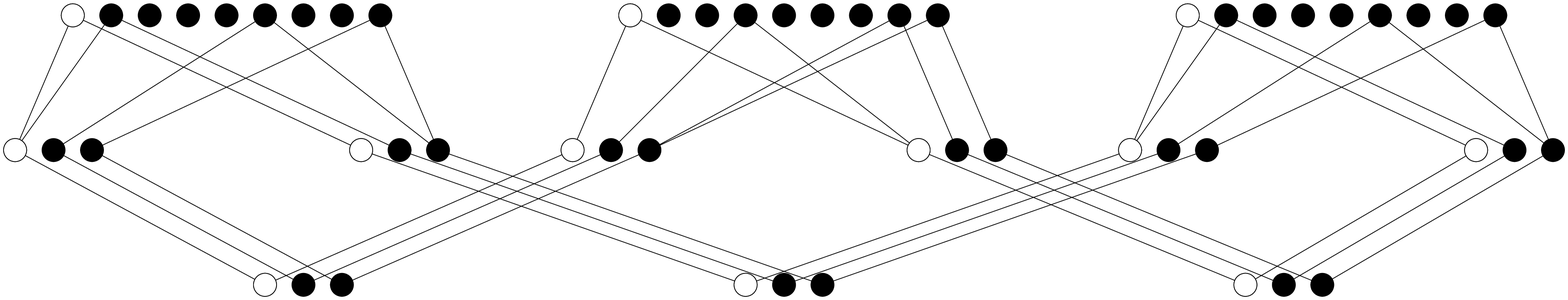}
\caption{The poset $\pp_{\phi}$ corresponding to $\phi \in \mathcal{S}$ in Example~\ref{ex:degreenphard}.}
\label{fig:pdegreenphard}
\end{figure}}

\longversion{\figpdegreenphard}

We now formalize the ideas outlined above. Let $\phi(x_1,\ldots,x_n)=\delta_1 \wedge \cdots \wedge \delta_m$ be in $\mathcal{S}$.  
For $j \in [n]$ and $i \in [m]$, we write $x_j \in \delta_i$ if a literal on variable $x_j$ occurs in clause $\delta_i$, 
and we let $\textup{var}(\delta_i)=\{ x_j \mid j \in [n], x_j \in \delta_i \}$.  
For all $i \in [m]$, let $(g_{i,1},\ldots,g_{i,7})$ be a fixed ordering of the assignments in $\{0,1\}^{\textup{var}(\delta_i)}$ satisfying $\delta_i$, 
and let $(i_1,i_2,\ldots,i_{m-1})=(1,\ldots,i-1,i+1,\ldots,m)$. We define our two posets $\qq_\phi$ and $\pp_\phi$ below.  

The poset $\qq_\phi$ has universe $Q_\phi=Q_{0} \cup Q_{1} \cup Q_{2}$, where 
\begin{align*}
Q_{0} = & \{ c_{(i,j)}, c_{(i,m)}, c_{(m,j)} \mid i, j \in [m-1], i\neq j \}\text{,} \\
Q_{1} = & \{ f_{(i,j)} \mid i,j \in [m], i\neq j \} \text{,}\\
Q_{2} = & \{ d_{(i,j)} \mid 1 \leq i<j \leq m \} \text{,}
\end{align*}
and its cover relation is defined by the following:
\begin{itemize}
\item[(E1)] $f_{(i,j)},f_{(j,i)} \prec^{\qq_\phi} d_{(i,j)}$ for all $1 \leq i<j \leq m$.  
\item[(E2)] For all $i \in [m]$, 
\longshort{\begin{align*}
f_{(i,i_1)} & \succ^{\qq_\phi} c_{(i,i_1)} \prec^{\qq_\phi} f_{(i,i_2)} \succ^{\qq_\phi} \cdots \succ^{\qq_\phi} c_{(i,i_{m-1})} \prec^{\qq_\phi} f_{(i,i_{m-1})}\text{.}
\end{align*}}{\begin{align*}
f_{(i,i_1)} & \succ^{\qq_\phi} c_{(i,i_1)} \prec^{\qq_\phi} f_{(i,i_2)} \succ^{\qq_\phi} \cdots \\
\cdots & \succ^{\qq_\phi} c_{(i,i_{m-1})} \prec^{\qq_\phi} f_{(i,i_{m-1})}\text{.}
\end{align*}}
\end{itemize}

The poset $\pp_\phi$ has universe $P_\phi=P_{0} \cup P_{1} \cup P_{2}$ where, 
\begin{align*}
P_0 = & \{ c_{(i,j),a},c_{(i,m),a},c_{(m,j),a} \mid i,j \in [m-1], i\neq j, a \in [7] \}\text{,}\\
P_1 = & \{ f_{(i,j),a} \mid i,j\in [m], i\neq j, a \in [7] \}\text{,}\\
P_2 = & \{ d_{(i,j),(a,a')} \mid 1 \leq i<j \leq m, (a,a') \in [7]^2 \}\text{,}
\end{align*}
and its cover relation is defined by the following:
\begin{itemize}
\item[(D1)] For all $1 \leq i<j \leq m$, it holds that 
$f_{(i,j),a},f_{(j,i),a'} \prec^{\pp_\phi} d_{(i,j),(a,a')}$ if and only if $g_{i,a}(x)=g_{j,a'}(x)$ 
for all $x \in \textup{var}(\delta_i) \cap \textup{var}(\delta_j)$.  

\item[(D2)] For all $i \in [m]$ and $a \in [7]$, 
\longshort{\begin{align*}
f_{(i,i_{1}),a} & \succ^{\pp_\phi} c_{(i,i_1),a} \prec^{\pp_\phi} f_{(i,i_2),a} \succ^{\pp_\phi} \cdots \succ^{\pp_\phi} c_{(i,i_{m-1})} \prec^{\pp_\phi} f_{(i,i_{m-1}),a}\text{.}
\end{align*}}{\begin{align*}
f_{(i,i_{1}),a} & \succ^{\pp_\phi} c_{(i,i_1),a} \prec^{\pp_\phi} f_{(i,i_2),a} \succ^{\pp_\phi} \cdots\\ 
\cdots & \succ^{\pp_\phi} c_{(i,i_{m-1})} \prec^{\pp_\phi} f_{(i,i_{m-1}),a}\text{.}
\end{align*}}
\end{itemize}
Since $\textup{cover\textup{-}degree}(\pp_{\phi}) \leq 1+7=8$ and 
$\textup{depth}(\pp_{\phi}) \leq 3$, $\mathcal{P}_{\textup{degree}}=\{ \pp_{\phi} \mid \phi \in \mathcal{S} \}$ 
has bounded degree by Proposition~\ref{pr:diagram}.

\longshort{\begin{theorem}}{\begin{theorem}[$\star$]}
\label{th:degreenphard} 
$\textsc{Emb}(\mathcal{P}_{\textup{degree}})$ is $\textup{NP}$-hard. 
\end{theorem}

\newcommand{\pfdegreenphard}[0]{
\begin{proof}
We give a polynomial-time many-one reduction from the satisfiability problem over 
$\mathcal{S}$ to the problem $\textsc{Emb}(\mathcal{P}_{\textup{degree}})$, 
which suffices since the source problem is $\textup{NP}$-hard.  

The reduction maps an instance $\phi \in \mathcal{S}$ of the satisfiability problem, 
say $\phi(x_1,\ldots,x_n)=\delta_1 \wedge \cdots \wedge \delta_m$,  
to the instance $(\mathbf{Q}_{\phi},\mathbf{P}_{\phi})$ of $\textsc{Emb}(\mathcal{P}_{\textup{degree}})$.  
The reduction is clearly polynomial-time computable.  

For correctness, let $g \colon \{x_1,\ldots,x_n\} \to \{0,1\}$ be an assignment satisfying $\phi$.  
Recall that $(g_{i,1},\ldots,g_{i,7})$ is a fixed ordering of the assignments in $\{0,1\}^{\textup{var}(\delta_i)}$ satisfying $\delta_i$, 
for all $i \in [m]$.  Let $(a_1,\ldots,a_m) \in [7]^m$ be such that 
$g|_{\textup{var}(\delta_i)}=g_{i,a_i}$ for all $i \in [m]$.  It is easy to check 
that the function $e \colon Q_{\phi} \to P_{\phi}$ defined by setting:
\begin{itemize}
\item $e(c_{(i,j)})=c_{(i,j),a_i}$ for all $c_{(i,j)} \in Q_0$;
\item $e(f_{(i,j)})=f_{(i,j),a_i}$ for all $f_{(i,j)} \in Q_1$;
\item $e(d_{(i,j)})=d_{(i,j),(a_i,a_j)}$ for all $d_{(i,j)} \in Q_2$;
\end{itemize}
embeds $\qq_{\phi}$ into $\pp_{\phi}$. 

Conversely, let $e \colon Q_{\phi} \to P_{\phi}$ embed $\qq_{\phi}$ into $\pp_{\phi}$.  
We show that $\phi$ is satisfiable.  Note that $e(Q_i) \subseteq P_i$ for all $i \in \{0,1,2\}$, 
because $e$ maps all $3$-element chains in $\qq_{\phi}$ into $3$-element chains in $\pp_{\phi}$, 
all $3$-element chains in $\qq_{\phi}$ link three elements in $Q_0$, $Q_1$, and $Q_2$, in this order, 
and all $3$-element chains in $\pp_{\phi}$ link three elements in $P_0$, $P_1$, and $P_2$, in this order.  

We first claim that for all $i \in [m]$, 
there exists exactly one $a \in [7]$ such that, 
for all $j \in [m]\setminus \{i\}$, 
it holds that $e(f_{(i,j)})=f_{(i,j),a}$.  Assume for 
a contradiction that $e(f_{(i,j)})=f_{(i,j),a}$ and $e(f_{(i,j')})=f_{(i,j'),a'}$ 
for some $i \in [m]$, $a \neq a' \in [7]$, and $j \neq j' \in [m]\setminus\{i\}$; 
without loss of generality, let $j<j'$.  By (E2), 
$f_{(i,j)}$ reaches $f_{(i,j')}$ through a fence of length 
$2(j'-j)$, starting in $Q_1$ and alternating steps in $Q_0$ and $Q_1$; but by (D2), 
$f_{(i,j),a}$ does not reach $f_{(i,j'),a'}$ through a fence of length $2(j'-j)$, 
starting in $P_1$ and alternating steps in $P_0$ and $P_1$, contradicting the assumption that $e$ is an embedding.

Let $(a_1,\ldots,a_m) \in [7]^m$ be uniquely determined by the previous 
claim.  We now claim that, for all $i,j \in [m]$ such that $i \neq j$, 
and all $x \in \textup{var}(\delta_i) \cap \textup{var}(\delta_{j})$, 
it holds that $g_{i,a_i}(x)=g_{j,a_j}(x)$.  Assume without loss of generality that $i<j$.  
By (E1), $f_{(i,j)},f_{(j,i)} \prec^{\qq_\phi} d_{(i,j)}$.  
By hypothesis, $e(f_{(i,j)})=f_{(i,j),a_i}$ and $e(f_{(j,i)})=f_{(j,i),a_j}$.  
Therefore, since $e$ is an embedding, $f_{(i,j),a_i},f_{(j,i),a_j} \prec^{\pp_\phi} e(d_{(i,j)})$; 
thus, by (D1), $e(d_{(i,j)})=d_{(i,j),(a_i,a_j)}$ that is, 
$g_{i,a_i}(x)=g_{j,a_j}(x)$ for all $x \in \textup{var}(\delta_i) \cap \textup{var}(\delta_{j})$.  

By the above, $g=g_{1,a_1} \cup \cdots \cup g_{m,a_m}$ is a 
function from $\{x_1,\ldots,x_n\}$ to $\{0,1\}$.  Since $g_{i,a_i}$ 
satisfies $\delta_i$ for all $i\in [m]$, it follows that $g$ satisfies $\phi$, concluding the proof.
\end{proof}}

\longversion{\pfdegreenphard}

\longshort{\subsection{Isomorphism in Polynomial Time on Bounded Width Posets}\label{sect:wdtract}}{\subsection{Isomorphism in Polytime on Bounded Width Posets}\label{sect:wdtract}}

The insight on bounded width used to prove tractability of the embedding problem 
essentially scales to the isomorphism problem.

\begin{theorem}\label{th:isoptime}
Let $\mathcal{P}$ be a class of 
posets of bounded width.  Then, 
$\textsc{Iso}(\mathcal{P})$ is polynomial-time tractable.
\end{theorem}
\begin{proof}
The proof utilizes three known facts from the literature.  

Let $\mathbf{R}$ be any poset.  
For all $S \subseteq R$, let 
$(S]$ be \emph{downset} generated by $S$ in  $\mathbf{R}$, i.e., 
 $(S]=\{ r \in R \mid \exists s \in S \text{ such that $r \leq^{\mathbf{R}} s$}\}$.  
Let $l(\mathbf{R})$ be the order defined 
by equipping the universe of all antichains in $\mathbf{R}$ 
by the relation $A \leq^{l(\mathbf{R})} A'$ if and only if 
$(A]\subseteq (A']$.  Note that, if $\textup{width}(\mathbf{R})$ 
is considered as a constant, the construction of $l(\mathbf{R})$ 
is polynomial-time computable from $\mathbf{R}$. 

The three needed facts are the following.  First, for any (finite) poset $\mathbf{R}$, 
the structure $l(\mathbf{R})$ is a (finite) distributive lattice \cite[Proposition~5.5.5]{Schroder03}.  
Second, the substructure of $l(\mathbf{R})$ 
generated by join irreducible elements is isomorphic to $\mathbf{R}$ \cite[Theorem~5.5.6]{Schroder03}; 
recall that, if $\mathbf{L}=(L,\leq)$ is a lattice, 
then $j \in L$ is \emph{join irreducible} if, for all $l,l' \in L$, 
if $j$ is the least upper bound of $l$ and $l'$, then $j=l$ or $j=l'$.   
Third, the isomorphism problem restricted to finite distributive lattices 
is polynomial-time tractable \cite{GorazdIdziak95}.  

Using the previous facts, we design the following algorithm.
Let $w \in \mathbb{N}$ be the upper bound on the width of posets in $\mathcal{P}$.  
Let $(\mathbf{Q},\mathbf{P})$ be an instance of $\textsc{Iso}(\mathcal{P})$.  Let $|P|=n$.  
If $|Q| \neq n$, or $\textup{width}(\mathbf{Q})>w$, 
or $\textup{width}(\mathbf{Q}) \neq \textup{width}(\mathbf{P})$, then reject; 
the condition is checkable in time $O(w \cdot n^2)$ by Theorem~\ref{th:felsner}.  
Otherwise, in polynomial time, 
compute $l(\qq)$ and $l(\pp)$ 
and accept if and only if $l(\qq)$ and $l(\pp)$ are isomorphic.  

The algorithm clearly runs in polynomial time.  For correctness, 
notice that $\qq$ and $\pp$ are isomorphic if and only if 
$l(\qq)$ and $l(\pp)$ are isomorphic.  For the nontrivial direction (backwards), 
if $f$ is an isomorphism from $l(\qq)$ to $l(\pp)$, 
then let $f'$ be the restriction of $f$ to the join irreducible elements 
of $l(\qq)$.  It is easy to check that $f'$ is bijective into 
the join irreducible elements of $l(\pp)$, hence, 
using the second fact mentioned above, 
$f'$ is an isomorphism between $\qq$ and $\pp$.
\end{proof}

\section{Conclusion}\label{sect:concl}

We embarked on the study of the model checking problem 
on 
posets; compared to graphs, 
the problem is largely unexplored, and we made a first contribution 
by studying basic syntactic fragments (existential logic) 
and fundamental poset invariants (including width, depth, and degree).  
Our complexity classification for existential logic also carries over to the \emph{jump number} 
(between size and width in Figure~\ref{fig:overwparcompl}); a future direction is to extend our study to \emph{dimension} 
(above width \cite{CaspardLeclercMonjardet12} and degree \cite{FurediKahn86} in Figure~\ref{fig:overwparcompl}).
%
 
Our main algorithmic result, fixed-parameter tractability of existential logic on bounded width posets, 
raises the natural question of whether model checking the full first-order logic 
is fixed-parameter tractable on classes of posets of bounded width. 
We propose this as a topic for future research.  

%
%

%

%
%
%

\shortversion{\acks

This research was supported by ERC Starting Grant (Complex Reason, 239962) 
and FWF Austrian Science Fund (Parameterized Compilation, P26200).}


\bibliographystyle{abbrvnat}


\longversion{\newpage}

\end{document}